%% file: main.tex
\documentclass[letterpaper]{article} 
\usepackage{aaai25}  
\usepackage{times}  
\usepackage{helvet}  
\usepackage{courier}  
\usepackage[hyphens]{url}  
\usepackage{graphicx} 
\usepackage{natbib}  
\usepackage{caption} 
\usepackage{algorithm}
\usepackage{algorithmic}

\usepackage{newfloat}
\usepackage{listings}
\DeclareCaptionStyle{ruled}{labelfont=normalfont,labelsep=colon,strut=off} 
\floatstyle{ruled}
\newfloat{listing}{tb}{lst}{}
\floatname{listing}{Listing}
\title{Reducing Leximin Fairness to Utilitarian Optimization}
\author {
    Eden Hartman\textsuperscript{\rm 1},
    Yonatan Aumann\textsuperscript{\rm 1},
    Avinatan Hassidim\textsuperscript{\rm 1,2},
    Erel Segal-Halevi\textsuperscript{\rm 3}
}
\affiliations {
    \textsuperscript{\rm 1}Bar-Ilan University, Israel\\
    \textsuperscript{\rm 2}Google, Israel
\\
    \textsuperscript{\rm 3}Ariel University, Israel
\\
    eden.r.hartman@gmail.com, aumann@cs.biu.ac.il, avinatan@cs.biu.ac.il, erelsgl@gmail.com
}

\usepackage{amsmath,amsthm,amssymb}
\usepackage{xspace}

\usepackage{cleveref}

\newtheorem{theorem}{Theorem}[section]
\newtheorem{lemma}[theorem]{Lemma}
\newtheorem{corollary}[theorem]{Corollary}
\newtheorem{observation}[theorem]{Observation}
\newtheorem{claim}[theorem]{Claim}
\newtheorem*{theorem*}{Theorem}
\newtheorem*{lemma*}{Lemma}

\newtheorem{definition}{Definition}[section]
\newtheorem{example}{Example}[section]

\newcommand\proofsketchname{Proof sketch} 
\newenvironment{proofsketch}%
   {\renewcommand\proofname\proofsketchname%
    \begin{proof}}%
   {\end{proof}}

\newcommand{\Hquad}{\hspace{0.5em}} 

\newcommand{\feasP}[1]{F(\text{#1})}
\newcommand{\objPofV}[2]{obj(\text{#1},#2)}

\newcommand{\objCx}[1]{\objPofV{\ref{eq:compact-OP}}{#1}}
\newcommand{\feasC}{\feasP{\ref{eq:compact-OP}}}

\newcommand{\worldStates}{S} 

\newcommand{\dist}{$\worldStates$-distribution\xspace}

\newcommand{\leximinPreferred}{\succ}
\newcommand{\weaklyPreferred}{\succeq}
\newcommand{\multApprox}{\alpha}
\DeclareMathOperator*{\argmax}{arg\,max}

\usepackage[most]{tcolorbox}
\makeatletter
\newcommand*{\rom}[1]{\expandafter\@slowromancap\romannumeral #1@}
\makeatother

\newcommand{\progCompact}{P1}
\newcommand{\progAppFirst}{P2}
\newcommand{\progAppSecond}{P3}
\newcommand{\dualApp}{D3}

\newcommand{\expectedValBy}[2]{E_{#1}^{\uparrow}(#2)}
\newcommand{\forXinY}[2]{\forall #1 \in [#2]}
\newcommand{\ztCons}{z_t}

\newcommand{\shallowX}{X_{\leq\multApprox}}

\newcommand{\x}{\mathbf{x}}
\newcommand{\xt}{\mathbf{x'}}
\newcommand{\retSol}{\x^n}

\newcommand{\appendixName}[2]{#2}

\newcommand{\shallowXg}{X_{\leq\multApprox}}
\usepackage{array}

\newcommand{\equWithExp}[2]{#1\text{:}\\ & #2}
\newcommand{\vv}{\mathbf{v}}

\newcommand{\orderedVby}[1]{v^{\uparrow}_{#1}}
\newcommand{\multError}{\beta}

\ifdefined\DRAFT
	\newcommand{\er}[1]{\textcolor{blue}{#1}}
	\newcommand{\erel}[1]{\er{(Erel says: #1)}}
	\newcommand{\eden}[1]{\textcolor{red}{(Eden says: #1)}}
	
	\newcommand{\rmark}[1]{\textcolor{BrickRed}{ #1}}
	\newcommand{\yon}[1]{\textcolor{Green}{(Yonatan says: #1)}}

\else
	\newcommand{\er}[1]{#1}
	\newcommand{\erel}[1]{}
	\newcommand{\eden}[1]{}
	
	\newcommand{\yon}[1]{}
	\newcommand{\rmark}[1]{#1}
\fi

\begin{document}

\maketitle

\begin{abstract}
Two prominent objectives in social choice are \emph{utilitarian} - maximizing the sum of agents' utilities, and \emph{leximin} - maximizing the smallest agent's utility, then the second-smallest, etc. Utilitarianism is typically computationally easier to attain but is generally viewed as less fair.  This paper presents a general reduction scheme that, given a utilitarian solver, produces a distribution over states (deterministic outcomes) that is leximin in expectation. 
Importantly, the scheme is robust in the sense that, given an \emph{approximate} utilitarian solver, it produces a lottery that is approximately-leximin (in expectation) - with the same approximation factor.  We apply our scheme to several social choice problems: stochastic allocations of indivisible goods, giveaway lotteries, and fair lotteries for participatory budgeting.
\end{abstract}

%
\textcolor{white}{
.\\}

\section{Introduction}
In social choice, the goal is to find the best choice for society, but 'best' can be defined in many  ways. 
Two frequent, and often contrasting definitions are the \emph{utilitarian best}, which focuses on maximizing the total welfare (i.e., the sum of utilities); and the \emph{egalitarian best}, which focuses on maximizing the least  utility. 
The \emph{leximin best} generalizes the egalitarian one. It first aims to maximize the least utility; then, among all options that maximize the least utility, it chooses the one that maximizes the second-smallest utility, among these --- the third-smallest utility, and so forth. Leximin is often the solution of choice in social choice applications, and frequently used (e.g., \citet{freeman2019indv, bei_truthful_2022, Cheng2023, flanigan2024}). 

\paragraph{Calculating the Optimal Choice.}
Calculating a choice that maximizes utilitarian welfare is often easier than finding one that maximizes egalitarian welfare, while finding one that is leximin optimal is typically even more complex. 
For example, when allocating indivisible goods among agents with additive utilities, finding a choice (in this case, an allocation) that maximizes the utilitarian welfare can be done by greedily assigning each item to the agent who values it most. Finding an allocation that maximizes the egalitarian welfare, however, is NP-hard \cite{Bansal2006santaClaus}, even in this relatively simple case.



In this paper, we show that knowing how to efficiently maximize the utilitarian welfare is sufficient in order to find a fair leximin solution.
\begin{figure*}[t]
\centering
\includegraphics{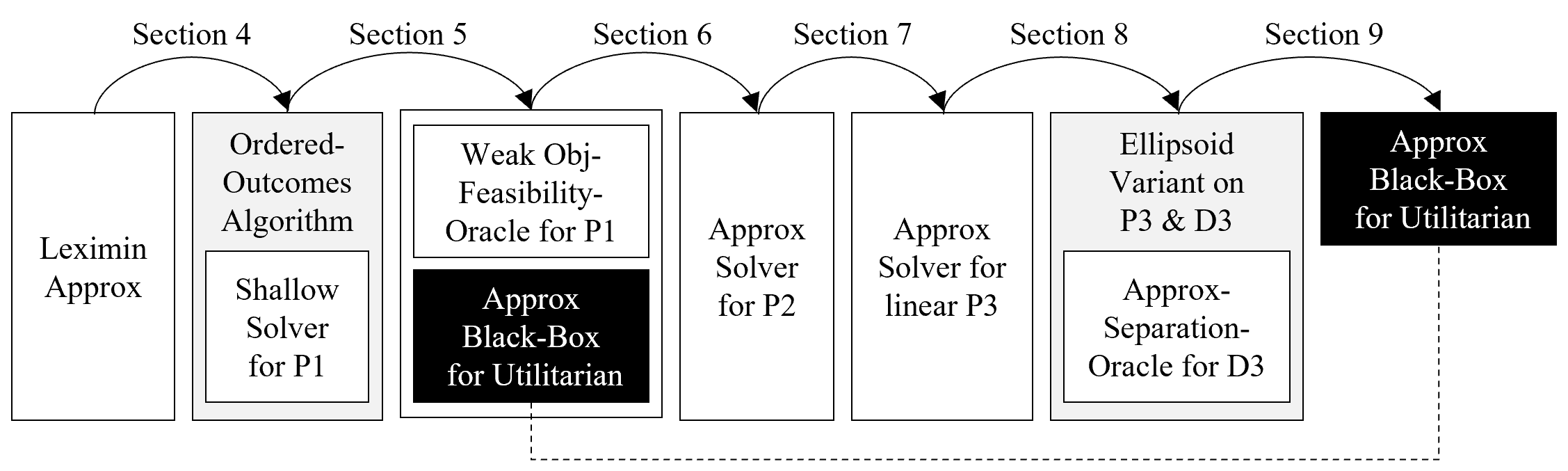} 
\caption{High level description of the reduction algorithm.
An arrow from element A to B denotes that the corresponding section reduces problem A to B.
White components are implemented in this paper; gray 
components represent existing algorithms; the black component is the black-box for the utilitarian welfare.
}
\label{fig:reduction-high-level-des}
\end{figure*} 

\paragraph{Contributions.}
The core contribution of this paper is a general protocol that, when provided with a procedure for optimizing the utilitarian welfare (for a given problem), outputs a solution that optimizes the expected leximin welfare (for the same problem).
By \emph{expected} leximin we mean a distribution over deterministic solutions, for which the expectations of the players' utilities is leximin optimal. 
Crucially, our protocol \emph{extends to approximations}, in the following sense: given an \emph{approximate} solver for the utilitarian welfare, the protocol outputs a solutions that approximates the expected leximin optimal, and the same approximation factor is preserved.
In all, with our protocol at hand, optimizing expected leximin welfare is no more difficult than optimizing utilitarian welfare.

We demonstrate the significance of this reduction by applying it to three social choice problems as follows. 

First, we consider the classic problem of \emph{allocations of indivisible goods}, where one seeks to fairly distribute a a set of indivisible goods among a set of heterogeneous agents. 
Maximizing the utilitarian welfare in this case is well-studied. 
Using our reduction, the previously mentioned greedy algorithm for agents with additive utilities, allows us to achieve a leximin optimal lottery over the allocations in polynomial time. For submodular utilities, approximating leximin to a factor better than $(1-\frac{1}{e})$ is NP-hard. However, by applying our reduction, existing approximation algorithms for utilitarian welfare can be leveraged to prove that a $0.5$-approximation can be obtained deterministically, while the best-possible $(1-\frac{1}{e})$-approximation can be obtained with high probability.

Second, we consider the problem of \emph{giveaway lotteries} \cite{arbiv_fair_2022}, where there is an event with limited capacity and groups wish to attend, but only-if they can all be admitted together. 
Maximizing the utilitarian welfare in this setting can be seen as a knapsack problem, for which there is a well-known FPTAS (fully polynomial-time approximation scheme). 
Using our reduction, we obtain an FPTAS for leximin as well.

Lastly, we consider the problem of \emph{fair lotteries for participatory budgeting}, where one seeks to find a fair lottery over the possible budget allocations. 
When agents have additive utilities, 
maximizing the utilitarian welfare can also be interpreted as a knapsack problem (albeit in a slightly different way), which allows us to provide an FPTAS for leximin.

\paragraph{Organization.}
Section \ref{sec:model} introduces the model and required definitions.

In Sections \ref{sec:main-loop}-\ref{sec:proof-main-thm}, we prove our main result: an algorithm for finding a leximin-approximation using an approximate black-box for utilitarian welfare, while importantly, preserving the same approximation factor.~
The reduction is done step by step.
Section \ref{sec:main-loop} reduces the problem of leximin-approximation to another problem;
then, each Section $k\in \{5,6,7,8,9\}$ 
reduces the problem introduced in Section $k-1$ to another problem,
where in Section~\ref{sec:separation-oracle} the reduced problem is approximate utilitarian optimization.~
Section \ref{sec:proof-main-thm} ties the knots to prove the entire reduction, and extends the result to randomized solvers.~
A schematic description of the reduction structure is provided in Figure \ref{fig:reduction-high-level-des}.

The applications are shown in Section \ref{sec:apps}. 
Lastly, Section \ref{sec:conclusions-and-future} concludes with some future work directions.
Most proofs are deferred to the appendix.


\subsection{Related Work}
Recently, there has been a wealth of research focused on finding a fair leximin lottery for specific problems. Examples include algorithms proposed for Representative Cohort Selection \cite{henzinger_et_al_FORC_2022}, Giveaway Lotteries \cite{arbiv_fair_2022}, Allocating
Unused Classrooms \cite{Kurokawa2018leximinRealWorld}, and Selecting Citizens’
Assemblies \cite{flanigan2021fair}.
This paper, in contrast, provides a general protocol that can be applied to a wide range of problems.

Alongside these, many papers describe general algorithms for exact leximin optimization 
\citep{Ogryczak_1997_loc, Ogryczak2004TelecommunicationsND, Ogryczak_2006}. 
These algorithms usually rely on a solver for single-objective problem, which, in our context, is NP-hard.
Recently, \citet{hartman2023leximin} adapted one of these algorithms to work with approximately-optimal solver. However, designing such a solver remains quite challenging. 
Our work generalizes these approaches by proving that this algorithm still functions with even a weaker type of solver, which we show can be implemented for many problems.

Another significant area of research focuses on leximin approximations. Since leximin fairness involves multiple objectives simultaneously, it is not straightforward to define what leximin approximation is. 
Several definitions have been proposed. The definition we employ is related to that of  \cite{Kleinberg20012Routing,abernethy2024lexicographic}. 
Other definitions can be found in~\cite{hartman2023leximin, henzinger_et_al_FORC_2022}.
An extensive comparison of the different definitions, including examples, is provided in Appendix \appendixName{\ref{apx:def-lex-approx}}{H}.

The work closest to ours is 
\cite{kawase_max-min_2020}, which laid the foundation for this research. 
Their paper studies the problem of stochastic allocation of indivisible goods (see Section \ref{sec:app-stoc-alloc} for more details), and proposes a reduction from egalitarian welfare to utilitarian welfare for this specific problem. 
We extend their work in two ways. 
First, we extend their approach from the allocation problem to any problem where lotteries make sense.  
Second, we show that a black-box for the utilitarian welfare is even more powerful, as it can also be used for leximin, rather than the egalitarian welfare.


\section{Preliminaries}
We denote the set $\{1,\dots, n\}$ by $[n]$ for $n \in \mathbb{N}$.

\paragraph{Mathematical Programming.}
Throughout the paper, we frequently use mathematical programming to define optimization problems. A program is characterized by the following three elements. 
(1) Constraints: used to define the set of feasible solutions, which forms a subset of $\mathbb{R}^m$ for some $m \in \mathbb{N}$. 
(2) Type: the program can be either a maximization or minimization program. 
(3) Objective function: assigns an objective value to each feasible solution.
The goal is to find a feasible solution with an optimal objective value (either maximal or minimal, depending on the problem type).

\paragraph{Notations.}
For a given program, P, the \emph{feasible region}, denoted by $\feasP{P}$, represents the set of vectors that satisfy all the constraints of P.
We say that a vector $\mathbf{v} \in \feasP{P}$ is \emph{a feasible solution for P} (interchangeably: a solution for P or feasible for P), and denote its objective value according to the objective function of P by $\objPofV{P}{\mathbf{v}}$.

\section{Model and Definitions}\label{sec:model}
The setting postulates a set of $n$ agents $N = \{1,\ldots,n\}$, and a set of deterministic options $\worldStates$ --- this set represents the possible deterministic allocations - in the fair division setting, or the possible budget allocations - in a budgeting application. 
For simplicity from now on, we refer to $\worldStates$ as \emph{states} and number them  $\worldStates = \{s_1, \ldots, s_{|S|}\}$. 

We seek solutions that are \emph{distributions} over states.  Formally, an  \emph{\dist} is a probability distribution over the set $\worldStates$, and $X$ is the set of all such distributions: 
\begin{align*}
    X = \{\x = (x_1, \ldots, x_{|S|}) \in \mathbb{R}^{|S|}_{\geq 0} \mid \Hquad \sum_{j=1}^{|S|} x_j = 1 \}
\end{align*}
Importantly, we allow the number of states, $|S|$, to be exponential in $n$. This implies that even describing a solution requires exponential time. 
To address this, we follow \citet{kawase_max-min_2020} and assume that solutions are represented in \emph{sparse form} — a list of indices of entries with positive values, together with their corresponding values.
We will see that our algorithms return \emph{sparse solutions} — that is, the number of states with positive probability will be polynomial in $n$. These solutions can be efficiently described in polynomial time in their sparse form.

\begin{definition}[A Poly-sparse Vector]\label{def:poly-sized-sol}
A vector, $\mathbf{v} \in \mathbb{R}^{m}_{\geq 0}$ for some $m \in \mathbb{N}$, is a \emph{poly-sparse} if no more than polynomial number (in $n$) of its entries are non-zero. 
\end{definition}

\paragraph{Degenerate-State.}
We assume that there exists a single degenerate-state, $s_d \in S$, that gives all the agents utility $0$.
This will allow us to also consider sub-probabilities over the other states, which give positive utility to some agents. 
The use of degenerate-state makes sense in our setting, as we assume all utilities are non-negative.

\paragraph{Utilities.} The utility of agent $i \in N$ from a given state is described by the function $u_i \colon \worldStates \to \mathbb{R}_{\geq 0}$, which is provided as a value oracle.\footnote{This means that the algorithm does not have a direct access to the utility function.  Rather, given $s \in S$, the value $u_i(s)$ can be obtained in $\mathcal{O}(1)$ for any $i = 1, \ldots, n$.}

The utility of agent $i$ from a given \dist, $\mathbf{x}\in X$, is the expectations
$
	E_i(\mathbf{x}) = \sum_{j=1}^{|S|} x_j \cdot u_i(s_j).
$
The vector of expected utilities of all agents from a solution $\x$ is denoted by $\mathbf{E}(\mathbf{x}) = (E_1(\mathbf{x}), \ldots, E_n(\mathbf{x}))$; and referred to as the expected vector of $\mathbf{x}$.

\subsection{Leximin Fairness}\label{sec:leximin-def}
We aim for a leximin-optimal \dist: one that maximizes the smallest expected-utility, then, subject to that, maximizes the second-smallest expected-utility, and so on. It is formally defined below.

\paragraph{The Leximin Order.} For $\mathbf{v}\in \mathbb{R}^n$, let $\mathbf{v}^{\uparrow}$ be the corresponding vector sorted in non-decreasing order, and $v^{\uparrow}_i$ the $i$-th smallest element (counting repetitions). For example, if $\mathbf{v}=(1,4,7,1)$ then $\mathbf{v}^{\uparrow}=(1,1,4,7)$, and $v^{\uparrow}_3=4$. 

For $\mathbf{v},\mathbf{u}\in \mathbb{R}^n$, we say that $\mathbf{v}$ is (weakly) \emph{leximin-preferred} over $\mathbf{u}$, denoted $\mathbf{v} \weaklyPreferred \mathbf{u}$, if one of the following holds. 
Either $\mathbf{v}^{\uparrow} = \mathbf{u}^{\uparrow}$ (in which case they are \emph{leximin-equivalent}, denoted $\mathbf{v} \equiv \mathbf{u}$).
Or there exists an integer $1 \leq k\leq n $ such that $v^{\uparrow}_i  = u^{\uparrow}_i$ for $i < k$, and $v^{\uparrow}_k > u^{\uparrow}_k$
(in which case $\mathbf{v}$ is \emph{strictly} leximin-preferred over $\mathbf{u}$, denoted $\mathbf{v} \leximinPreferred \mathbf{u}$).

Note that the $\mathbf{E}(\mathbf{x})$'s are $n$-tuples, so the leximin order applies to them. 
\begin{observation}\label{obs:leximin-order-total}
    Let $\mathbf{v},\mathbf{u}\in \mathbb{R}^n$. Exactly one of the following holds: either $\mathbf{v} \weaklyPreferred \mathbf{u}$ or $\mathbf{u} \leximinPreferred \mathbf{v}$.
\end{observation}

\paragraph{Leximin Optimal.} We say that $\mathbf{x}^* \in X$ is a \emph{leximin-optimal \dist} if 
$\mathbf{E}(\mathbf{x}^*) \weaklyPreferred
\mathbf{E}(\mathbf{x})$
for all $\mathbf{x}\in X$. 

\paragraph{Approximation.}
Throughout the paper, we denote by $\multApprox \in (0,1]$
a multiplicative approximation ratio.

\newcommand{\xA}{\x^{A}}
\begin{definition}[Leximin-Approximation]\label{def:leximin-approx-weak}
    We say that an \dist, $\xA \in X $, is 
an \emph{$\multApprox$-leximin-approximation}
if 
$\mathbf{E}(\xA)
\weaklyPreferred
\multApprox \cdot \mathbf{E}(\mathbf{x})$ 
for all $\mathbf{x}\in X$.
\end{definition}


\begin{observation}\label{claim:1-approx-is-opt}
    An \dist is a $1$-leximin-approximation if-and-only-if it is leximin-optimal.
\end{observation}

\subsection{Utilitarian Optimization}\label{sec:utilitarian-def}
\paragraph{Utilitarian Optimal.} We say that $\mathbf{x}^{uo} \in X$ is a \emph{utilitarian-optimal \dist} if it maximizes the sum of expected utilities: 
\begin{align*}
     \forall \mathbf{x} \in X \colon \quad \sum_{i=1}^n E_i(\mathbf{x}^{uo}) \geq \sum_{i=1}^n E_i(\mathbf{x})
\end{align*}

\paragraph{Stochasticity is Unnecessary.}
In fact, for utilitarian welfare, there always exists a deterministic solution -- a single state -- that maximizes the sum of expected utilities.
%
%
\begin{lemma}\label{lemma:stoc-unnecessary-for-utilitarian}
Let 
$
\displaystyle
    s^{uo}\in \argmax_{s \in S} \sum_{i=1}^n u_i(s).
$
Then 
\begin{align*}
     \forall \mathbf{x} \in X \colon \quad \sum_{i=1}^n u_i(s^{uo}) \geq \sum_{i=1}^n E_i(\mathbf{x}).
\end{align*}
\end{lemma}
Algorithms for utilitarian welfare are typically designed for this deterministic setting (where the goal is to find a utilitarian-optimal \emph{state}). Our reduction requires a deterministic solver.

\subsubsection{Utilitarian Solver.}
The proposed reduction requires a utilitarian welfare solver that is robust to scaling each utility function by a different constant.
Formally:

\begin{tcolorbox}[left=2pt, right=2pt,  
colback=black!5!white,colframe=black!50!black, colbacktitle=black!75!black,
title=\textbf{(\rom{1})~$\multApprox$-Approximate Black-Box for Utilitarian Welfare}]
  \textbf{Input:} $n$ non-negative constants $c_1, \ldots, c_n$.
  \tcblower
  \textbf{Output:} A state, $s^{uo} \in S$, for which:
  \begin{align*}
      \forall s \in S \colon \quad \sum_{i =1}^n c_i\cdot u_i(s^{uo}) \geq \multApprox \sum_{i =1}^n  c_i\cdot u_i(s).
  \end{align*}
\end{tcolorbox}
When $\multApprox = 1$, we say that we have an \emph{exact} black-box.

\paragraph{}
Many existing solvers for the utilitarian welfare are inherently robust to scaling the utility functions by a constant, as they handle a class of utilities that are closed under this operation. For instance, in the division of goods, various classes of utilities, such as additive and submodular, are closed under constant scaling.

At times, however, rescaling can result in diverging from the problem definition. For example, if the problem definition assumes that utilities are normalized to sum up to $1$, then the definition is not robust to rescaling. In such cases, technically, a polynomial solver need not be able to solve the scaled version. 
Hence, we explicitly include the assumption that the solver is robust to rescaling.

\section{Main Loop}\label{sec:main-loop}
The main algorithm used in the reduction is described in Algorithm \ref{alg:basic-ordered-Outcomes}.
It is adapted from the Ordered Outcomes algorithm of~\cite{Ogryczak_2006} for finding an exact leximin-optimal solution. It uses a given solver for the maximization 
program below. 

\paragraph{The Program \ref{eq:compact-OP}.}
The program is parameterized by an integer $t \in N$ and $t-1$ constants $(z_1, \ldots, z_{t-1})$; its only variable is $\mathbf{x}$ (a vector of size $|S|$ representing an \dist):
\begin{align*}
    \max &&& \sum_{i=1}^{t}\expectedValBy{i}{\mathbf{x}} - \sum_{i=1}^{t-1}  z_i \tag{\progCompact}\label{eq:compact-OP}\\
    s.t. &&& (\text{\progCompact.1}) 
    \Hquad  
    \sum_{j=1}^{|S|} x_j = 1
    \nonumber\\
    &&& (\text{\progCompact.2}) 
    \Hquad  
    x_j \geq 0 && j = 1, \ldots, |S|
    \nonumber\\
                    &&& (\text{\progCompact.3}) \Hquad \sum_{i=1}^{\ell}\expectedValBy{i}{\mathbf{x}} \geq \sum_{i=1}^{\ell}  z_i && \forXinY{\ell}{t-1} \nonumber
\end{align*}
Constraints  (\progCompact.1--2)  simply ensure that  $\mathbf{x}\in X$.
Constraint (\progCompact.3) says that for any $\ell<t$, the sum of the smallest $\ell$ expected-utilities is at least the sum of the $\ell$ constants $z_1,\ldots,z_{\ell}$.
The objective of a solution\footnote{
See Preliminaries for more details.} $\x \in \feasC$ is the difference between the sum of its smallest $t$ expected-utilities and the sum of the $t-1$ constants $z_1,\ldots,z_{t-1}$. 
%

\paragraph{Algorithm \ref{alg:basic-ordered-Outcomes}.}
The algorithm has $n$ iterations. In each iteration $t$, it uses the given solver for \ref{eq:compact-OP} with the following parameters: the iteration counter $t$, and $(z_1,\ldots,z_{t-1})$ that were computed in previous iterations.
The solver returns a solution for this \ref{eq:compact-OP}, denoted by $\x^t$. If $t<n$ then its objective value, denoted by $\ztCons$, is used in the following iterations as an additional parameter.
%
%
Finally, the solution $\mathbf{x}^n$ generated at the last application of the solver is returned by the main loop.

Notice that the program evolves in each iteration. For any $t \in [n-1]$, the program at iteration $t+1$ differs from the program at iteration $t$ in two ways: first, the objective function changes; and second, an additional constraint is introduced as part of Constraint (\progCompact.3), $\sum_{i=1}^{t}\expectedValBy{i}{\mathbf{x}} \geq \sum_{i=1}^{t}  z_i$ (for $\ell = t$). 
This constraint is equivalent to:
$\sum_{i=1}^{t}\expectedValBy{i}{\mathbf{x}} - \sum_{i=1}^{t-1}z_i  \geq z_t$, which essentially ensures that any solution for following programs achieves an objective value at-least $\ztCons$ according to the objective function of the program at iteration $t$. In other words, this constraint guarantees that as we continually improving the situation.
This implies, in particular, that $\x^t$ remains feasible for the $(t+1)$-th program.

\begin{observation}\label{obs:xt-feasible-for-t+1}
    Let $t \in [n-1]$. The solution obtained in the $t$-th iteration of \Cref{alg:basic-ordered-Outcomes}, $\x^t$, is also feasible for the $(t+1)$-th program.
\end{observation}


\begin{algorithm}[t]
\caption{Main Loop}
\label{alg:basic-ordered-Outcomes}
\textbf{Input}: A solver for \ref{eq:compact-OP}.
\begin{algorithmic}[1] 
\FOR{$t=1$ to $n$}
\STATE Let $\x^t$ be the solution returned by applying the given solver
with parameters $t$ and $(z_1,\ldots,z_{t-1})$.
\STATE Let $\ztCons := \objCx{\x^t}$.
\ENDFOR
\STATE \textbf{return} $\mathbf{x}^n$.
\end{algorithmic}
\end{algorithm}

\paragraph{Solving \ref{eq:compact-OP} Exactly.}
By
\citet{Ogryczak_2006} (Theorem 1), the returned $\x^n$ is leximin-optimal
when the solver for \ref{eq:compact-OP} is \emph{exact} - that is, it returns a solution $\x^t \in \feasC$ with optimal objective value --- $\objCx{\x^t} \geq \objCx{\x}$ for any solution $\x \in \feasC$.

However, in some cases, no efficient exact solver for \ref{eq:compact-OP} is known. An example is stochastic allocation of indivisible goods among agents with submodular utilities, described in Section \ref{sec:app-stoc-alloc}  --- in this case, it is NP-hard even for $t=1$, as its optimal solution maximizes the egalitarian welfare (see \cite{kawase_max-min_2020} for the hardness proof).

\paragraph{Solving \ref{eq:compact-OP} Approximately.}
Our initial attempt to deal with this issue was to follow the approach of \citet{hartman2023leximin}. They consider an \emph{approximately-optimal} solver for \ref{eq:compact-OP}, 
which returns a solution $\x^t \in \feasC$ with approximately-optimal objective value ---  $\objCx{\x^t} \geq \multApprox \cdot  \objCx{\x}$ for any solution $\x \in \feasC$.
%

When $t=1$, the algorithm of \citet{kawase_max-min_2020} is an approximate-solver for \ref{eq:compact-OP}. However, for $t > 1$, their technique no longer works due to fundamental differences in the structure of the resulting programs.
%
%
Designing such a solver is very challenging; all our efforts to design such a solver for several NP-hard problems were unsuccessful.

\paragraph{Solving \ref{eq:compact-OP} Shallowly.}
Our first contribution is to show that Algorithm \ref{alg:basic-ordered-Outcomes} can work with an even weaker kind of solver for \ref{eq:compact-OP}, that we call a \emph{shallow solver}. 
The term "shallow" is used since the solver returns a solution whose objective value is optimal only with respect to a \emph{subset} of the feasible solutions, as described below.


Recall that $s_d$ is the degenerate-state that gives utility $0$ to all agents, and $x_d$ is its probability according to $\mathbf{x}$. 
We consider the set of solutions that use at-most $\multApprox$ fraction of the distribution for the other states, which give positive utility to some agents:
\begin{align*}
    X_{\leq \alpha} = \{\x \in X \mid \sum_{j \neq d} x_j \leq \alpha\}.
\end{align*}
%
A \emph{shallow-solver} is defined as follows: 




\begin{tcolorbox}[left=2pt, right=2pt,  
colback=black!5!white,colframe=black!50!black, colbacktitle=black!75!black,
title=\textbf{(\rom{2})~$\multApprox$-Shallow-Solver for \ref{eq:compact-OP}}]
  \textbf{Input:}~ An integer $t \in N$ and rationals $z_1,\ldots, z_{t-1}$.
  %
  \tcblower
  \textbf{Output:} A solution
  $\x^t \in \feasC$ such that \\$\objCx{\x^t} \geq \objCx{\x}$ for any $\x \in \feasC \cap \shallowX$.
\end{tcolorbox}

In words: the solver returns a solution $\x^t \in \feasC$ whose objective value is guaranteed to be optimal comparing only to solutions that are also in $\shallowX$. 
%
%
This is in contrast to an exact solver, where the objective value of the returned solution is optimal comparing to all solutions.\footnote{See Table \appendixName{\ref{tab:3-solvers}}{1} in Appendix \appendixName{\ref{apx:using-shallow-solver}}{B} for comparison of the solvers.}
Clearly, when $\multApprox=1$, we get an exact solver, as $X_{\leq 1} = X$.

Notice that $\x^t$ does not required to be in $\shallowX$, so its objective value might be \emph{strictly-higher} than the optimal objective value of the set $\feasC \cap \shallowX$.


\begin{lemma}\label{lemma:alg1-shallow-solver}
    Given an $\multApprox$-shallow-solver for \ref{eq:compact-OP}, Algorithm \ref{alg:basic-ordered-Outcomes} returns an $\multApprox$-leximin-approximation. 
\end{lemma}

\begin{proofsketch}
Consider the solution $\mathbf{x}^n$  returned by Algorithm \ref{alg:basic-ordered-Outcomes}. As $\mathbf{x}^n$ is feasible for the program solved in the last iteration, and as constraints are only being added along the way, we can conclude that $\mathbf{x}^n$ is feasible for all of the $n$ programs solved during the algorithm run.

Next, we suppose by contradiction that $\mathbf{x}^n$ is \emph{not} an $\multApprox$-leximin-approximation.
By definition, this means that there exists an $\x \in X$ such that $\multApprox \mathbf{E}(\x) \succ \mathbf{E}(\mathbf{x}^n)$.
That is, there exists an integer $k \in [n]$ such that: $\multApprox\expectedValBy{i}{\x} = \expectedValBy{i}{\mathbf{x}^n}$ for $i < k$ and $\multApprox\expectedValBy{k}{\x} > \expectedValBy{k}{\mathbf{x}^n}$.

We then construct $\xt$ from $\x$  as follows: $x'_j := \multApprox \cdot x_j$ for any $j \neq d$, and $x'_d = 1- \multApprox \sum_{j\neq d }x_j$. It is easy to verify that $\xt \in \shallowX$ and that $E_i(\xt) = \multApprox E_i(\x)$ for any $i$. This means that $\expectedValBy{i}{\xt} = \expectedValBy{i}{\mathbf{x}^n}$ for $i < k$ (as $\multApprox\expectedValBy{i}{\x} = \expectedValBy{i}{\mathbf{x}^n}$).

Next, we consider the program solved in iteration $t=k$.
As $\retSol$ satisfies all its constraints, and as these constraints impose a lower bound on the sum of the least $\ell < k$ expected-utilities, we can conclude that $\xt$ satisfies all these constraints as well. Thus, $\xt \in \feasC \cap \shallowX$.

Finally, we prove that the objective-value of $\xt$ for this program is strictly-higher than the one obtained by the shallow-solver (i.e., $z_k$) --- in contradiction to the solver guarantees.
\end{proofsketch}

\section{A Shallow-Solver for \ref{eq:compact-OP}}\label{sec:designing-shallow-solver}
Our next task is to design a shallow solver for \ref{eq:compact-OP}.
We use a \emph{weak objective-feasibility-oracle}, defined as follows:

\begin{tcolorbox}[left=2pt, right=2pt,  
colback=black!5!white,colframe=black!50!black, colbacktitle=black!75!black,
title=\textbf{(\rom{3})~$\multApprox$-Weak Objective-Feasibility-Oracle for \ref{eq:compact-OP}}]
  \textbf{Input:} ~
  An integer $t \in N$ and rationals $z_1,\ldots, z_{t-1}$,\\ and another rational $\ztCons$.
  \tcblower
  \textbf{Output:} 
  One of the following claims regarding $\ztCons$:\vspace{0.5em}
    \begin{minipage}{0.27\linewidth}
  \setlength{\baselineskip}{1em}
\textbf{Feasible }
\end{minipage}%
\hfill
\begin{minipage}{0.73\linewidth}
$\exists \x \in \feasC$ s.t. $\objCx{\x} \geq \ztCons$.\\
In this case, the oracle returns such $\mathbf{x}$.
\end{minipage}\\[0.15em]

  \begin{minipage}{0.27\linewidth}
  \setlength{\baselineskip}{1em}
\textbf{Infeasible } \\
\textbf{Under-$\shallowX$ }
\end{minipage}%
\hfill
\begin{minipage}{0.73\linewidth}
$\nexists \mathbf{x} \in \feasC \cap \shallowX$ s.t. $\objPofV{\ref{eq:compact-OP}}{\x} \geq \ztCons$.
\end{minipage}

\end{tcolorbox}

Note that these claims are not mutually exclusive, as $\ztCons$ can satisfy both conditions simultaneously. In this case, the oracle may return any one of these claims.

\begin{lemma}
\label{lemma:red-shallow-solver-to-feas-test}
Given an $\multApprox$-weak objective-feasibility-oracle for \ref{eq:compact-OP} (\rom{3}), an $\multApprox$-approximate black-box for the utilitarian welfare (\rom{1}),
and an arbitrary vector in $\feasC$.~~ An efficient $\multApprox$-shallow-solver for \ref{eq:compact-OP} (\rom{2}) can be designed.
\end{lemma}

\begin{proofsketch}
The solver is described in \appendixName{\Cref{alg:shallow-solver}}{Algorithm 2} in \appendixName{\Cref{apx:desiging-shallow}}{Appendix C}.
We perform binary search over the potential objective-values $\ztCons$ for the program \ref{eq:compact-OP}.

As a lower bound for the search, we simply use the objective value of the given feasible solution.

As an upper bound for the search, we use an upper bound on the utilitarian welfare of the original utilities $u_i$ -- it is sufficient as the objective function is the sum of the $t$-smallest utilities minus a positive constant; thus, an upper bound on the sum of all utilities can bound the maximum objective value as well.
We obtain this upper bound by using the given 
$\multApprox$-approximate black-box with $c_i = 1$ for all $i \in N$; and then take $\frac{1}{\multApprox}$ times the returned value.
 
During the binary search, we query the weak objective-feasibility-oracle about the current value of $\ztCons$.
If the oracle asserts Feasible, we increase the lower bound; otherwise, we decrease the upper bound.

We stop the binary search once we reach the desired level of accuracy.%
\footnote{For simplicity, we assume the binary search error is negligible, as it can be reduced below $\epsilon$ in time $\mathcal{O} (\log \frac{1}{\epsilon})$, for any $\epsilon > 0 $; the full proof in Appendix \appendixName{\ref{apx:desiging-shallow}}{C}, omits this assumption.}
Finally, we return the solution $\x$ corresponding to the highest $\ztCons$ for which the oracle returned Feasible. 

By the definition of the oracle, this $\x$ is feasible for \ref{eq:compact-OP}.
In addition, this $\ztCons$ is at-least as high as the optimal objective across $\shallowX$ -- as all the values higher than $\ztCons$ that were considered by the algorithm, were determined as Infeasible-Under-$\shallowX$.
This gives us a shallow solver for \ref{eq:compact-OP}.
\end{proofsketch}

\section{Weak Objective-Feasibility-Oracle for \ref{eq:compact-OP}}
To design a weak objective-feasibility-oracle for \ref{eq:compact-OP},
we modify 
\ref{eq:compact-OP} as follows. 
First, we convert the optimization program \ref{eq:compact-OP},
to a feasibility program (without an objective), by adding a constraint saying that the objective function of \ref{eq:compact-OP} is at least the given constant $\ztCons$: $\sum_{i=1}^{t}\expectedValBy{i}{\mathbf{x}} \geq \sum_{i=1}^{t}  z_i$.
We then make two changes to this feasibility program: (1) remove Constraint (\progCompact.1) that ensures that the sum of values is $1$,
and (2) add the objective function: $\min \sum_{j=1}^{|S|} x_j$.
We call the resulting program \ref{eq:min-sum-OP}:
\begin{align*}
    &\min \quad \sum_{j=1}^{|S|} x_{j} \quad s.t. \tag{\progAppFirst}\label{eq:min-sum-OP}\\
    & (\text{\progAppFirst.1}) \Hquad x_{j} \geq 0 &&  j = 1, \ldots, |S| \nonumber\\
    & (\text{\progAppFirst.2}) \Hquad \sum_{i=1}^{\ell}\expectedValBy{i}{\mathbf{x}} \geq \sum_{i=1}^{\ell}  z_i && \forXinY{\ell}{t} \nonumber
\end{align*}
Note that (\progAppFirst.2) contains the constraints (\progCompact.3), as well as the new constraint added when converting to a feasibility program.

As before, the only variable is $\mathbf{x}$. However, in this program, 
a vector $\mathbf{x}$ can be feasible without being in $X$, as its elements are not required to sum to $1$.

We shall now prove that a weak objective-feasibility-oracle for \ref{eq:compact-OP} can be designed given a solver for \ref{eq:min-sum-OP}, which returns a poly-sparse\footnote{A poly-sparse vector (def. \ref{def:poly-sized-sol}) is one whose number of non-zero values can be bounded by a polynomial in $n$.} approximately-optimal solution. As 
\ref{eq:min-sum-OP} is a \emph{minimization} program and as $\multApprox\in(0,1]$, it is defined as follows:

\newcommand{\xRet}{\x^A}
\newcommand{\xRetj}{x^A_j}


\begin{tcolorbox}[left=2pt, right=2pt,  
colback=black!5!white,colframe=black!50!black, colbacktitle=black!75!black,
title=\textbf{(\rom{4})~$\frac{1}{\multApprox}$-Approx.-Optimal-Sparse-Solver for \ref{eq:min-sum-OP}}]
  \textbf{Input:} ~
  An integer $t \in N$ and rationals $z_1,\ldots, z_{t}$.
  \tcblower
  \textbf{Output:} A poly-sparse solution $\xRet \in \feasP{\ref{eq:min-sum-OP}}$ such that $\objPofV{\ref{eq:min-sum-OP}}{\xRet} \leq \frac{1}{\multApprox} \objPofV{\ref{eq:min-sum-OP}}{\x}$ 
for any $\mathbf{x} \in \feasP{\ref{eq:min-sum-OP}}$.
\end{tcolorbox}

\begin{lemma}\label{lemma:red-feas-test-to-approx-P2}
Given an $\frac{1}{\multApprox}$-approximately-optimal-sparse-solver for \ref{eq:min-sum-OP} (\rom{4}),
~~ an $\multApprox$-weak objective-feasibility-oracle for \ref{eq:compact-OP}  (\rom{3}) can be obtained.
\end{lemma}


\begin{proof}
We apply the solver to get a vector $\xRet \in \feasP{\ref{eq:min-sum-OP}}$, and then check whether $\objPofV{\ref{eq:min-sum-OP}}{\xRet} := \sum_{j=1}^{|S|}\xRetj$ is at-most $1$.
Importantly, this can be done in polynomial time, since $\xRet$ is a poly-sparse vector. 

If so, we assert that $\ztCons$ is Feasible.
Indeed, while $\xRet$ may not be feasible for \ref{eq:compact-OP}, we can construct another poly-sparse vector, from $\xRet$, that is feasible for \ref{eq:compact-OP} and achieves an objective value of at-least $\ztCons$.
The new vector, $\xt$, is equal to $\xRet$ except that $x'_d := (1- \sum_{j=1}^{|S|} \xRetj)$.
It can be easily verified that the new vector, $\xt$, is in $\feasC$, and that $\objCx{\xt}\geq\ztCons$, as required.
    
Otherwise, $\sum_{j=1}^{|S|}\xRetj > 1$. In this case, we assert that $\ztCons$ is Infeasible-under-$\shallowX$. 
Indeed, assume for contradiction that there exists $\x \in \feasC \cap \shallowX$ with $\objCx{\x} \geq \ztCons$.
We show that this implies that the optimal (min.) objective-value for \eqref{eq:min-sum-OP} is at-most $\multApprox$, meaning that any $\frac{1}{\multApprox}$-approximation should yield an objective value of at-most $1$ -- in contradiction to the objective value of $\xRet$. 
We construct a new vector $\xt$ from $\x$, where $\xt$ equals to $\x$
except that $x'_d := 0$.
As $\x \in \feasC$ and $E_i(\x) = E_i(\xt)$ for all $i\in N$, it follows that $\xt \in \feasP{\ref{eq:min-sum-OP}}$.
But $\x \in \shallowX$, which means that it uses at-most $\multApprox$ for states $j \neq d$. 
This implies that $\objPofV{\ref{eq:min-sum-OP}}{\xt} := \sum_{j=1}^{|S|} x'_j $ is at-most $ \multApprox$, which ensures that the optimal (min.) objective is also at-most $\multApprox$ -- as required.
\end{proof}


\section{Approximately-Optimal Solver for \ref{eq:min-sum-OP}}\label{sec:moving-to-LP}
The use of $\expectedValBy{}{}$ operator makes both \ref{eq:compact-OP} and \ref{eq:min-sum-OP} non-linear.
However, \citet{Ogryczak_2006} showed that \ref{eq:compact-OP}  can be ``linearized`` by replacing the constraints using $\expectedValBy{}{}$ with a polynomial number of linear constraints.
We take a similar approach for \ref{eq:min-sum-OP} to construct the following \emph{linear} program \ref{eq:app-vsums-OP}:
%
%
\begin{align}
    &\min \quad \sum_{j=1}^{|S|} x_{j} \quad s.t. \tag{\progAppSecond}\label{eq:app-vsums-OP}\\
    & (\text{\progAppSecond.1}) \Hquad x_{j} \geq 0 &&  j = 1, \ldots, |S| \nonumber\\
    & (\text{\progAppSecond.2}) \Hquad \ell y_{\ell} - \sum_{i=1}^n m_{\ell,i}\geq \sum_{i=1}^{\ell}  z_i && \forXinY{\ell}{t} \nonumber \\
    & (\text{\progAppSecond.3}) \Hquad m_{\ell,i} \geq y_{\ell} - \sum_{j=1}^{|S|} x_j \cdot u_i(s_j)  && \forXinY{\ell}{t},\Hquad \forXinY{i}{n} \nonumber \\
    & (\text{\progAppSecond.4}) \Hquad m_{\ell,i} \geq 0  && \forXinY{\ell}{t},\Hquad \forXinY{i}{n} \nonumber
\end{align}
Constraints (\progAppSecond.2--4) introduce $t(n+1) \leq n(n+1)$ auxiliary variables:  $y_{\ell}$ and $m_{\ell,i}$ for all $\ell \in [t]$ and $ i\in [n]$.
%
%
\newcommand{\yRet}{\mathbf{y}^A}
\newcommand{\mRet}{\mathbf{m}^A}
We formally prove the equivalence between the two sets of constraints
in Appendix D; and show that it implies that the required solver for \ref{eq:min-sum-OP} can be easily derived from the same type of solver for \ref{eq:app-vsums-OP}.

\begin{tcolorbox}[left=2pt, right=2pt,  
colback=black!5!white,colframe=black!50!black, colbacktitle=black!75!black,
title=\textbf{(\rom{5})~$\frac{1}{\multApprox}$-Approx.-Optimal-Sparse-Solver for \ref{eq:app-vsums-OP}}]
  \textbf{Input:} ~
  An integer $t \in N$ and rationals $z_1,\ldots, z_{t}$.
  \tcblower
  \textbf{Output:} A poly-sparse $\left(\xRet, \yRet, \mRet\right) \in \feasP{\ref{eq:app-vsums-OP}}$ such that $\objPofV{\ref{eq:app-vsums-OP}}{\left(\xRet, \yRet, \mRet\right)} \leq \frac{1}{\multApprox} \objPofV{\ref{eq:app-vsums-OP}}{\left(\x, \mathbf{y}, \mathbf{m}\right)}$ 
\\for any $\left(\x, \mathbf{y}, \mathbf{m}\right) \in \feasP{\ref{eq:app-vsums-OP}}$.
\end{tcolorbox}


\begin{lemma}\label{lemma:red-approx-solver-for-P3-to-solver-for-P2}
Given an $\frac{1}{\multApprox}$-approximately-optimal-sparse-solver for \ref{eq:app-vsums-OP} (\rom{5}),
~~ an $\frac{1}{\multApprox}$-approximately-optimal-sparse-solver for \ref{eq:min-sum-OP} (\rom{4}) can be obtained.
\end{lemma}

\begin{proofsketch}
    The equivalence between the two sets of constraints says that $\left(\x, \mathbf{y}, \mathbf{m}\right) \in \feasP{\ref{eq:app-vsums-OP}}$ if-and-only-if $\x \in \feasP{\ref{eq:min-sum-OP}}$.
    Thus, given a poly-sparse $\frac{1}{\multApprox}$-approximately-optimal solution $\left(\xRet, \yRet, \mRet\right) \in \feasP{\ref{eq:app-vsums-OP}}$; then, $\xRet$ (which is clearly a poly-sparse vector), is a $\frac{1}{\multApprox}$-approximately-optimal solution for \ref{eq:min-sum-OP}.
\end{proofsketch}


\section{Approximately-Optimal Solver for  \ref{eq:app-vsums-OP}}
\label{sec:approx-the-linear-program}
\ref{eq:app-vsums-OP} is a linear program, but has more than $|S|$ variables.
Although $|S|$ (the number of states) may be exponential in $n$, \ref{eq:app-vsums-OP} can be approximated in polynomial time
using a variant of the ellipsoid method, similarly to \citet{karmarkar_efficient_1982}, as described bellow.
The method uses an approximate separation oracle for the dual of the linear program \ref{eq:app-vsums-OP}:
%
\begin{align}
& \max \quad  \sum_{\ell=1}^{t} q_{\ell} \sum_{i=1}^{\ell}z_i \quad s.t. \tag{\dualApp}\label{eq:dual-vsums-OP}\\
&\begin{aligned}
    & (\text{\dualApp.1}) \Hquad \sum_{i=1}^n u_i(s_j) \sum_{\ell=1}^{t} v_{\ell,i}\leq 1 && \forall j = 1, \ldots, |S| \nonumber \\
 & (\text{\dualApp.2}) \Hquad \ell q_{\ell} - \sum_{i=1}^n v_{\ell,i} \leq 0 && \forXinY{\ell}{t} \nonumber \\
& (\text{\dualApp.3}) \Hquad -q_{\ell} +v_{\ell,i} \leq 0 && \forXinY{\ell}{t},\Hquad \forXinY{i}{n} \nonumber \\
& (\text{\dualApp.4}) \Hquad q_{\ell} \geq 0  &&\forXinY{\ell}{t} \nonumber\\
& (\text{\dualApp.5}) \Hquad v_{\ell,i} \geq 0  && \forXinY{\ell}{t},\Hquad \forXinY{i}{n} \nonumber
\end{aligned}
\end{align}
Similarly to \ref{eq:app-vsums-OP}, the program \ref{eq:dual-vsums-OP} is parameterized by an integer $t$, and rational numbers $(z_1, \ldots, z_t)$.
It has a polynomial number of variables: $q_{\ell}$ and $v_{\ell,j}$ for any $\ell \in [t]$ and $j \in [n]$; and a potentially exponential number of constraints due to (\dualApp.1); see Appendix \appendixName{\ref{apx:primal-dual}}{E} for derivation.

We prove that the required solver for \ref{eq:app-vsums-OP} can be designed given the following procedure for its dual \ref{eq:dual-vsums-OP}:
\begin{tcolorbox}[left=2pt, right=2pt,  
colback=black!5!white,colframe=black!50!black, colbacktitle=black!75!black,
title=\textbf{(\rom{6})~$\frac{1}{\multApprox}$-Approx.-Separation-Oracle for \ref{eq:dual-vsums-OP}}]
  \textbf{Input:} ~
  An integer $t \in N$, rationals $z_1,\ldots, z_{t}$, and a potential assignment of the program variables $(\mathbf{q}, \mathbf{v})$.
  \tcblower
  \textbf{Output:} 
  One of the following regarding $(\mathbf{q}, \mathbf{v})$:\vspace{0.5em}
    \begin{minipage}{0.27\linewidth}
  \setlength{\baselineskip}{1em}
\textbf{Infeasible }
\end{minipage}%
\hfill
\begin{minipage}{0.73\linewidth}
At least one of the constraints is violated by $(\mathbf{q}, \mathbf{v})$.
In this case, the oracle returns such a constraint.
\end{minipage}\\[1em]

  \begin{minipage}{0.27\linewidth}
  \setlength{\baselineskip}{1em}
\textbf{$\frac{1}{\multApprox}$-Approx. } \\
\textbf{Feasible }
\end{minipage}%
\hfill
\begin{minipage}{0.73\linewidth}
All the constraints are $\frac{1}{\multApprox}$-approximately-maintained --- the left-hand side of the inequality is at least $\frac{1}{\multApprox}$ times the its right-hand side.
\end{minipage}
\end{tcolorbox}


In Appendix \appendixName{\ref{apx:ellipsoid}}{F}, we present the variant of the ellipsoid method that, given a $\frac{1}{\multApprox}$-approximate-separation oracle for the (max.) dual program, ~allows us to obtain a sparse \mbox{$\frac{1}{\multApprox}$-approximation} to the (min.) primal program (\appendixName{\Cref{lemma:ellipsoid-with-approx-oracle}}{Lemma F.1}). 
This allows us to conclude the following:
\begin{corollary}
\label{lemma:aso-to-approx-P3}
    Given a \emph{$\frac{1}{\multApprox}$-approximate-separation-oracle} for \ref{eq:dual-vsums-OP} (\rom{6}),~~~
    a $\frac{1}{\multApprox}$-approximately-optimal-sparse-solver for \ref{eq:app-vsums-OP} (\rom{5}) can be derived.
\end{corollary}



\section{Approx. Separation Oracle for \ref{eq:dual-vsums-OP}}\label{sec:separation-oracle}
Now, we design the required oracle 
using the given approximate black-box for the utilitarian welfare.
\begin{lemma}\label{lemma:sep-oracle-for-D}
\label{lemma:red-approx-P3-to-utiliterian}
    Given an $\multApprox$-approximate black-box for the utilitarian welfare (\rom{1}), ~~
    a $\frac{1}{\multApprox}$-approximate-separation-oracle for \ref{eq:dual-vsums-OP} (\rom{6}) can be constructed. 
\end{lemma}

\begin{proof}
    The oracle can be designed as follows.
    Given an assignment of the variables of the program \ref{eq:dual-vsums-OP}, Constraints (\dualApp.2--5) can be verified directly, as their number is polynomial in $n$.
    If a violated constraint was found, the oracle returns it.
    Otherwise, the potentially exponential number of constraints in (\dualApp.1) are treated collectively.
    We operate the approximate black-box with  $c_i := \sum_{\ell=1}^t v_{\ell,i}$ for $i \in N$, and obtain a state $s_k\in S$.
    If $\sum_{i=1}^n c_i \cdot u_i(s_k) > 1$, we declare that constraint $k$ in (\dualApp.1) is violated.
    Otherwise, we assert that the assignment is approximately-feasible.
    Indeed, in this case, Constraints (\dualApp.2)-(\dualApp.5) are exactly-maintained (since we first check them directly), and all the constraints in (\dualApp.1) are approximately-maintained, as proven below.
    By the definition of the approximate black-box, we can conclude that $\sum_{i=1}^n c_i \cdot u_i(s_k) \geq \multApprox \sum_i c_i \cdot u_i(s)$ for any state $s \in S$.
    It follows that, for any state $s \in S$, the corresponding constraint in (\dualApp.1) is approximately-maintained:
    \begin{align*}
        \sum_{i=1}^n u_i(s) \sum_{\ell=1}^t v_{\ell,i} \leq \frac{1}{\multApprox} \sum_{i=1}^n u_i(s_k) \sum_{\ell=1}^t v_{\ell,i} \leq \frac{1}{\multApprox} \cdot 1
    \end{align*}
\end{proof}

\section{The Main Result}
\label{sec:proof-main-thm}
Putting it all together, we obtain: 
\begin{theorem}\label{thm:maim}
Given an $\multApprox$-approximate black-box for the utilitarian welfare (\rom{1}).
~
An $\multApprox$-leximin-approximation (Def. \ref{def:leximin-approx-weak}) can be computed in time  polynomial in $n$ and the running time of the black-box. 
\end{theorem}

\begin{proof}

By \Cref{lemma:red-approx-P3-to-utiliterian}, given an $\multApprox$-approximate black-box for utilitarian welfare, 
a $\frac{1}{\multApprox}$-approximate-separation-oracle for \ref{eq:dual-vsums-OP} can be constructed.

By \Cref{lemma:aso-to-approx-P3}, 
using this oracle, we can construct a 
$\frac{1}{\multApprox}$-approximately-optimal-sparse-solver for \ref{eq:app-vsums-OP}.

By \Cref{lemma:red-approx-solver-for-P3-to-solver-for-P2},
we can use this solver to design a
$\frac{1}{\multApprox}$-approximately-optimal-sparse-solver for \ref{eq:min-sum-OP}.

By \Cref{lemma:red-feas-test-to-approx-P2},
this solver allows us to construct an $\multApprox$-weak objective-feasibility-oracle for \ref{eq:compact-OP}.

By \Cref{lemma:red-shallow-solver-to-feas-test}, a binary search with (1) this procedure, (2) the given $\multApprox$-approximate black-box for utilitarian welfare, and (3) an arbitrary solution for \ref{eq:compact-OP} (as follows);  allows us to design an $\multApprox$-shallow-solver for \ref{eq:compact-OP}.
As an arbitrary solution, for $t=1$, we take $\x^0$ defined as the \dist with $x^0_d=1$
and $x^0_j=0$ for $j\neq d$.
However, for $t \geq 2$, we rely on the fact that the shallow solver is used within the iterative  \Cref{alg:basic-ordered-Outcomes}, and take the solution returned in the previous iteration, $\x^{t-1}$, which, by \Cref{obs:xt-feasible-for-t+1}, is feasible for the program at iteration $t$ as well.

By \Cref{lemma:alg1-shallow-solver}, when this shallow solver is used inside 
\Cref{alg:basic-ordered-Outcomes}, the output is an 
$\multApprox$-leximin approximation.
\end{proof}
Together with \Cref{claim:1-approx-is-opt}, this implies:
\begin{corollary}\label{corollary:reduction-exact-case}
    Given an exact black-box for the utilitarian welfare, a leximin-optimal \dist can be obtained in polynomial time. 
\end{corollary}


\subsection*{Randomized Solvers.}\label{sec:randomized-solvers}
A \emph{randomized} black-box
returns a state that $\multApprox$-approximates the utilitarian welfare with probability $p>0$, otherwise returns an arbitrary state.
\Cref{thm:maim} can be extended as follows:
\begin{theorem}\label{lemma:randomized-solvers}
    Given a randomized $\multApprox$-approximate black-box for the utilitarian welfare with success probability $p \in (0,1)$. 
    ~~    An $\multApprox$-leximin-approximation
    can be obtained with the same success probability $p$ in time polynomial in $n$ and the running time of the black-box. 
\end{theorem}

\begin{proofsketch}
We first prove that the use of the randomized black-box does not effect feasibility --- that is, the output returned by Algorithm \ref{alg:basic-ordered-Outcomes} is always an \dist. Then, we prove the required guarantees regarding its optimally.
    
    Recall that the black-box is used only in two places --  to obtain an upper bound for the binary search over the potential
    objective-values for \ref{eq:compact-OP} (\Cref{sec:designing-shallow-solver}), and as part of the separation oracle for \ref{eq:dual-vsums-OP} (\Cref{sec:separation-oracle}).
    
    Inside the binary search, if the obtained value does not approximate the optimal utilitarian welfare, it might cause us to overlook larger possible objective values. 
    While this could affect optimality, it will not impact feasibility.
    
    As for the separation oracle, this change might cause us to determine that Constraint (\dualApp.1) is approximately-feasible even though it is not. However, in \appendixName{\Cref{sec:random-sep-oracle}}{Appendix F.1}, we prove that even with this change, the solution produced by the ellipsoid method remains feasible for the primal (although its optimality guarantees might no longer hold).  Consequently, the solution ultimately returned by \Cref{alg:basic-ordered-Outcomes} is also feasible (i.e., an \dist in $X$).

    We can now move on to the optimally guarantees. 
    Assume that the black-box for the utilitarian welfare is randomized and has a success probability $p<1$. 
    It is clear that iteratively calling this solver within the algorithm reduces the overall success probability. 
    
    To address this, we first boost the success probability of each iteration
    by calling the original black-box multiple times (on the same instance) and taking the best result among these calls. This boosts the success probability of the iteration (as failure corresponds to the probability that none of the calls to the original black-box was successful). 
    We prove that with an appropriate choice for the number of such repetitive calls, the total success probability of the entire algorithm can be as high as that of the original black-box ($p$), while maintaining polynomial complexity.
\end{proofsketch}

\section{Applications}\label{sec:apps}
This section provides three applications of our general reduction framework. 
Each application employs a different black-box for the associated utilitarian welfare.

\subsection{Stochastic Indivisible Allocations}\label{sec:app-stoc-alloc}
In the  problem of \textit{fair stochastic allocations of indivisible goods}, described by \citet{kawase_max-min_2020}, 
there is a set of $m$ indivisible goods, $G$, that needs to be distributed fairly among the $n$ agents. 

\paragraph{The Set $S$.} 
The states are the possible allocations of the goods to the agents. Each state can be described by a function mapping each good to the agent who gets it.
Accordingly, 
$S = \{ s \mid s \colon G \to N\}$,
and $|S|=n^m$.

We assume that agents care only about their own share, so we can abuse notation and let each $u_i$ take a bundle $B$ of goods.
The utilities are assumed to be normalized such that $u_i(\emptyset) = 0$, and monotone -- $u_i(B_1) \leq u_i(B_2)$ if $B_1 \subseteq B_2$.
Under these assumptions, different black-boxes for the utilitarian welfare exist.

\paragraph{The Utilitarian Welfare.} For any $n$ constants $c_1, \ldots, c_n$, the goal is to maximize the following:
\begin{align*}
    \max_{s \in S} \sum_{i=1}^n c_i \cdot u_i(s)
\end{align*}
Many algorithms for approximating utilitarian welfare are already designed for classes of utilities, which are closed under multiplication by a constant.
This means that, given such an algorithm for the original utilities $(u_i)_{i \in N}$, we can use it as-is for the utilities $(c_i \cdot u_i)_{i \in N}$.

\paragraph{Results.} 
%
When the utilities are additive,
it is known that a leximin-optimal solution can be found in polynomial time by introduce variables describing the allocation probability of each agent–item pair, and then solve it using an iterative linear-programming algorithm (e.g., \cite{willson, Ogryczak_2006}). However, this result can also be derived independently from our approach, since utilitarian welfare maximization is computationally easy—it can be solved in polynomial time by greedily assigning each item to the agent who values it most.
Together with Corollary \ref{corollary:reduction-exact-case}:
\begin{corollary}
    For additive utilities, a leximin-optimal \dist can be obtained in polynomial time. 
\end{corollary}

When the utilities are submodular,
approximating leximin to a factor better than $(1-\frac{1}{e})$ is NP-hard~\cite{kawase_max-min_2020}.\footnote{\citet{kawase_max-min_2020} prove that approximating the egalitarian welfare to a factor better than $(1-\frac{1}{e})$ is NP-hard. However, since an $\multApprox$-leximin-approximation is first-and-foremost an $\multApprox$-approximation to the egalitarian welfare, the same hardness result applies to leximin as well.}
However, as there is a deterministic $\frac{1}{2}$-approximation algorithm for the utilitarian welfare~\cite{Fisher1978}, by \Cref{thm:maim}:
\begin{corollary}
    For submodular utilities, a $\frac{1}{2}$-leximin-approximation can be found in polynomial time.
\end{corollary}
There is also a randomized $(1-\frac{1}{e})$-approximation algorithm for the case where utilities are submodular, with high success probability~\cite{vondrak_optimal_2008}. Thus, by \Cref{lemma:randomized-solvers}:
\begin{corollary}
    For submodular utilities, a $(1-\frac{1}{e})$-leximin-approximation can be obtained with high probability in polynomial time.
\end{corollary}



\subsection{Giveaway Lotteries}
In \textit{giveaway lotteries}, described by \citet{arbiv_fair_2022}, there is an event with a limited capacity, and groups who wish to attend it - but only-if all the members of the group can attend together.
Here, each group of people is an agent.
We denote the size of group $i$ by $w_i \in \mathbb{N}_{\geq 0}$ and the event capacity by $W \in \mathbb{N}_{\geq 0}$. It is assumed that $w_i \leq W$ for $i \in N$ and $\sum_{i\in N} w_i > W$.\footnote{\citet{arbiv_fair_2022} provide an algorithm to compute a leximin-optimal solution. However, their algorithm is polynomial only for a unary representation of the capacity.
}

\paragraph{The Set $S$.}
Each state describes a set of the groups that can attend the event together: $S = \{s \subseteq N \mid \sum_{i \in s} w_i \leq W\}$. Here, $|S|$ is only bounded by $2^n$.
%
%
%
%

The utility of group $i \in N$ from a state $s$ is $1$ if they being chosen according to $s$ (i.e., if $i \in s$) and $0$ otherwise.


\paragraph{The Utilitarian Welfare.}
For any $n$ constants $c_1, \ldots, c_n$, the goal is to maximize the following:
\begin{align*}
    \max_{s \in S} \sum_{i=1}^n c_i \cdot u_i(s) = \max_{s \in S} \sum_{i\in s}  c_i
\end{align*}
This is just a knapsack problem with $n$ item (one for each group), where the weights are the group sizes $w_i$ (as we only look at the legal packing $s\in S$), and the values are the constants $c_i$.

\paragraph{Result.} 
It is well known that there is an FPTAS for the Knapsack problem.  
By Theorem \ref{thm:maim}:
\begin{corollary}
    There exists an FPTAS for leximin for the problem of giveaway lotteries.
\end{corollary}



\subsection{Participatory Budgeting Lotteries}
The problem of \textit{fair lotteries for participatory budgeting}, was described by \citet{aziz2024participatoryBudgeting}. 
Here, the $n$ agents are voters, who share a common budget $B \in \mathbb{R}_{> 0}$ and must decide which projects from a set $P$ to fund.  
Each voter, $i \in N$, has an \emph{additive} utility over the set of projects, $u_i$; while the projects have costs described by  $cost \colon P \to \mathbb{R}_{>0}$.
\footnote{\citet{aziz2024participatoryBudgeting} study fairness properties based on \emph{fair share} and \emph{justified representation}.}



\paragraph{The set $S$.} 
The states are the subsets of projects that fit in the given budget: 
$S = \{s \subseteq P \mid cost(s) \leq B\}$.
The size of $S$ in this problem is only bounded by $2^{|P|}$.

\paragraph{The Utilitarian Welfare.} For any $n$ constants $c_1, \ldots, c_n$, the goal is to maximize the following:
\begin{align*}
    &\max_{s \in S} \sum_{i=1}^n c_i \cdot u_i(s) = \max_{s \in S} \sum_{i=1}^n \sum_{p \in s} c_i \cdot u_i(p)  && \text{(Additivity)}\\
    &=\max_{s \in S} \sum_{p \in s} \left(\sum_{i=1}^n c_i \cdot u_i(p) \right)
\end{align*}
This can also be seen as a knapsack problem where: the items are the projects, the weights are the costs, and the value of item $p \in P$ is $\sum_{i=1}^n c_i \cdot u_i(p)$.

\paragraph{Result.} 
As before, the existence of a FPTAS for the Knapsack problem together with Theorem \ref{thm:maim}, give:
\begin{corollary}
    There is an FPTAS for leximin for participatory budgeting lotteries.
\end{corollary}


\section{Conclusion and Future Work}\label{sec:conclusions-and-future}

In this work, we establish a strong connection between leximin fairness and utilitarian optimization, demonstrated by a reduction.
It is robust to errors in the sense that, given a black-box that approximates the utilitarian value, a leximin-approximation with respect to the same approximation factor can be obtained in polynomial time.

\paragraph{Negative Utilities.} Our current technique requires to assume that the utilities are non-negative since we use the degenerate-state. 
For non-positive utilities, we believe that a similar reduction might still be possible by redefining the degenerate state to represent a lower bound on the worst outcome with respect to negative utilities. For example, in chores allocation, an appropriate choice for the degenerate state can be an artificial allocation in which every agent is assigned all tasks.
We note that combining positive and negative utilities poses even more challenges -- as even the meaning of multiplicative approximation  in this case is unclear.


\paragraph{Nash Welfare.} The Nash welfare (product of utilities) is another prominent objective in social choice, offering a compelling compromise between the efficiency of utilitarianism and the fairness of egalitarianism. From a computational perspective, maximizing Nash welfare is typically as challenging as maximizing egalitarian welfare. An interesting question is whether a similar reduction can be constructed from Nash welfare (in expectation) to utilitarian optimization.

\paragraph{Applications.} We believe this method can also be applied to a variety of other problems, such as: Selecting a Representative Committee \cite{henzinger_et_al_FORC_2022}, Allocating Unused Classrooms \cite{Kurokawa2018leximinRealWorld}, Selecting Citizens’ Assemblies \cite{flanigan2021fair}, Cake-Cutting \cite{aumann2012computing,Edith21}, Nucleolus \cite{Elkind09}.

\paragraph{Best-of-both-worlds.} \citet{aziz2024participatoryBudgeting}, who study participatory budgeting lotteries, focus on fairness that is both ex-ante (which is a guarantee on the distribution) and ex-post (which is a guarantee on the deterministic support). In this paper, we guarantee only ex-ante fairness. Can ex-post fairness be achieved alongside it?
We note that our method ensures both ex-ante and ex-post efficiency -- since leximin guarantees Pareto-optimality with respect to the expected utilities $E_i$, any state in the support is Pareto-optimal with respect to the utilities $u_i$.

\paragraph{Deterministic Setting.} When the objective is to find a leximin-optimal (or approximate) \emph{state} rather than a distribution, it remains an open question whether a black-box for utilitarian welfare can still contribute, even if for a different approximation factor.

\paragraph{Truthfulness.} \citet{arbiv_fair_2022}, who study giveaway lotteries, prove that a leximin-optimal solution is not only fair but also truthful in this case. Can our leximin-approximation be connected to some notion of truthfulness?

\paragraph{Leximin-Approximation Definitions.} This paper suggests a weaker definition of leximin-approximation than the one proposed by \citet{hartman2023leximin} (see Appendix H for more details). Can similar results be obtained with the stronger definition?


\section*{Acknowledgments}
This research is partly supported by the Israel Science Foundation grants 712/20, 1092/24, 2697/22, and 3007/24. 
We sincerely appreciate Arkadi Nemirovski, Yasushi Kawase, Nikhil Bansal, Shaddin Dughmi, Edith Elkind, and Dinesh Kumar Baghel for their valuable insights, helpful answers, and clarifications. 
We are also grateful to the following members of the stack exchange network for their very helpful answers to our technical questions:
Neal Young,
\footnote{
    \url{https://cstheory.stackexchange.com/questions/51206} and 
    \url{https://cstheory.stackexchange.com/questions/51003}
    and
    \url{https://cstheory.stackexchange.com/questions/52353}
}
1Rock,
\footnote{
    \url{https://math.stackexchange.com/questions/4466551}
}
Mark L. Stone,%
\footnote{
    \url{https://or.stackexchange.com/questions/8633}
    and 
    \url{https://or.stackexchange.com/questions/11007}
    and 
    \url{https://or.stackexchange.com/questions/11311}
}
Rob Pratt,
\footnote{
    \url{https://or.stackexchange.com/questions/8980}
}
mtanneau,
\footnote{
    \url{https://or.stackexchange.com/questions/11910}
}
and mhdadk.
\footnote{
    \url{https://or.stackexchange.com/questions/11308}
}
Lastly, we would like to thank the reviewers in AAAI 2025 and COMSOC 2025 for their helpful comments.

\bibliography{main}

\clearpage

\appendix
\section*{Appendices Outline.} \Cref{apx:stoc-unnecessary-for-utilitarian} provides the proof of \Cref{lemma:stoc-unnecessary-for-utilitarian} that says that stochasticity is unnecessary for
utilitarian welfare .

In \Cref{apx:using-shallow-solver}, we give a comprehensive analysis of Algorithm \ref{alg:basic-ordered-Outcomes} when using a \emph{shallow} solver for the program \ref{eq:compact-OP}, including the proof of \Cref{lemma:alg1-shallow-solver}. 

\Cref{apx:desiging-shallow} proves \Cref{lemma:red-shallow-solver-to-feas-test}, by describing the implementation of the required shallow solver. 

\Cref{apx:eqv-linear} focuses on the equivalence between the programs \eqref{eq:min-sum-OP} and \eqref{eq:app-vsums-OP} -- including proof of \Cref{lemma:red-approx-solver-for-P3-to-solver-for-P2}.

\Cref{apx:primal-dual} provides the primal-dual derivation of programs \ref{eq:app-vsums-OP} and \ref{eq:dual-vsums-OP}.

\cref{apx:ellipsoid} presents the variant of the ellipsoid method used in this paper. 

\Cref{apx:random} addresses the case where the black-box for utilitarian welfare is randomized, proving \Cref{lemma:randomized-solvers}. 

Lastly, \Cref{apx:def-lex-approx} provides an extensive comparison between the different types of definitions for leximin-approximation.

\input{appendices/utilitarian}

\input{appendices/using-shallow}

\input{appendices/feas-oracle}

\input{appendices/constraints-eqv}

\input{appendices/primal-dual}
\input{appendices/ellipsoid}

\input{appendices/randomized-black-box}
\input{appendices/approx-def-comp}

\end{document}

%% file: appendices/utilitarian.tex
\section{Stochasticity Is Unnecessary for Utilitarian Welfare ~~ (Proof of \Cref{lemma:stoc-unnecessary-for-utilitarian})}\label{apx:stoc-unnecessary-for-utilitarian}
Recall that the \Cref{lemma:stoc-unnecessary-for-utilitarian} says that: 
\begin{lemma*}
Let 
$
\displaystyle
    s^{uo}\in \argmax_{s \in S} \sum_{i=1}^n u_i(s).
$
Then 
\begin{align*}
     \forall \mathbf{x} \in X \colon \quad \sum_{i=1}^n u_i(s^{uo}) \geq \sum_{i=1}^n E_i(\mathbf{x}).
\end{align*}
\end{lemma*}

\begin{proof}
Let $\mathbf{x}^{uo}$ be an \dist that maximizes the sum of expected-utilities:
\begin{align*}
\textbf{x}^{uo} \in \argmax_{x \in X} \sum_{i=1}^n E_i(s).
\end{align*}
It follows that:
\begin{align*}
    &\sum_{i=1}^n u_i(s^{uo}) =   \sum_{j=1}^{|S|} x^{uo}_j \sum_{i=1}^n u_i(s^{uo}) && \text{(As $\sum_{j=1}^{|S|} x^{uo}_j = 1$)}\\
    & \geq \sum_{j=1}^{|S|} x^{uo}_j \sum_{i=1}^n u_i(s_j)&& \text{(By def. of $s^{uo}$)}\\
    & = \sum_{i=1}^n \sum_{j=1}^{|S|} x^{uo}_j \cdot u_i(s_j) = \sum_{i=1}^n E_i(\mathbf{x}^{uo}).
\end{align*}
\end{proof}

%% file: appendices/using-shallow.tex
\newcommand{\up}[1]{up(#1)}
\newcommand{\upS}{up}
\newcommand{\down}[1]{down(#1)}
\newcommand{\downS}{down}


\begin{table*}[htb]
    \centering
    \begin{tabular}{|>{\arraybackslash}m{0.3\textwidth}|>{\arraybackslash}m{0.3\textwidth}|>{\arraybackslash}m{0.3\textwidth}|}
    \hline
    \centering \textbf{Exact Solver} & \centering \textbf{$\multApprox$-Approx.Optimal Solver} & \centering \textbf{$\multApprox$-Shallow  Solver} \tabularnewline
    \centering \cite{Ogryczak_2006}& \centering \cite{hartman2023leximin} & \centering (This paper) \tabularnewline
    \hline
    \multicolumn{3}{|c|}{Returns $\x^t \in \feasC$ (i.e., a feasible solution for \ref{eq:general-comp}) such that}\tabularnewline
    \hline
      \centering $\objCx{\x^t} \geq \objCx{\x}$
    & \centering $\objCx{\x^t} \geq \multApprox \cdot \objCx{\x}$
    & \centering $\objCx{\x^t} \geq \objCx{\x}$
    \tabularnewline
     \centering for all $\x \in \feasC$.
    &\centering for all $\x \in \feasC$.
    &\centering for all $\x \in \feasC \cap \shallowXg$.
    \tabularnewline
    \hline
    \end{tabular}
    \caption{Comparing the three solvers for (\progCompact), where $\multApprox \in (0,1]$ is the approximation-factor.
    }
    \label{tab:3-solvers}  
\end{table*}

\section{Using a Shallow Solver ~(Including Proof of \Cref{lemma:alg1-shallow-solver})}\label{apx:using-shallow-solver}

This appendix provides an analysis of Algorithm \ref{alg:basic-ordered-Outcomes} when using a \emph{shallow} solver for the program \ref{eq:compact-OP}. 
We use a bit more compact representation of the program where the two constraints (\ref{eq:compact-OP}.1--2) are merged:
\begin{align*}
    \max &&& \sum_{i=1}^{t}\expectedValBy{i}{\mathbf{x}} - \sum_{i=1}^{t-1}  z_i \tag{\ref{eq:compact-OP}}\label{eq:general-comp}\\
    s.t. &&& (\text{\progCompact.1--2}) 
    \Hquad  
    \x \in X
    \nonumber\\
                    &&& (\text{\progCompact.3}) \Hquad \sum_{i=1}^{\ell}\expectedValBy{i}{\mathbf{x}} \geq \sum_{i=1}^{\ell}  z_i && \forXinY{\ell}{t-1} \nonumber
\end{align*}
Recall that this program is parameterized by an integer $t \in N$ and $t-1$ constants $(z_1, \ldots, z_{t-1})$; its  only variable is $\x$ (a vector of size $|S|$ representing an \dist).

\subsection{The Three Different Solver Types}
\paragraph{Comparing the Types.} In Section \ref{sec:main-loop}, three types of solvers for  \ref{eq:general-comp}  were mentioned - exact, approximately-optimal, and shallow. 
Table \ref{tab:3-solvers} provides a comparison between the three types. 
We denote the solution returned by the solver by $\x^t$,
and recall that $\feasC$ is the set of feasible solutions for \ref{eq:compact-OP} (i.e., those who satisfy all its constraints); and that $\objCx{\x}$ describes the objective value of such a solution $\x \in \feasC$.

\paragraph{Exact Solver.} An exact solver returns a solution, $\x^t \in \feasC$,
whose objective value is maximum among the objective values of all the solutions; formally,
$\objCx{\x^t} \geq \objCx{\x}$ for any solution $\x \in \feasC$.

\paragraph{Approximately-Optimal Solver.} An approximately-optimal solver returns a solution, $\x^t \in \feasC$, whose objective value is \emph{at least $\multApprox$ times the} maximum among the objective values of all the solutions; equivalently, $\objCx{\x^t} \geq \multApprox \cdot \objCx{\x}$ for any solution $\x \in \feasC$.

\paragraph{Shallow Solver.} In contrast, a \emph{shallow} solver returns a solution, $\x^t \in \feasC$, whose objective value is at least the maximum among the objective values of a specific \emph{subset} of solutions — specifically,  the solutions in $\shallowXg$; formally, $\objCx{\x^t} \geq \objCx{\x}$ for any solution $\x \in \feasC \cap \shallowXg $.

The name 'shallow' aims to reflect that this solver only considers a subset of the entire solution space.
It is important to note that the objective value of the solution returned by a shallow solver might be \emph{strictly-higher} than the maximum within $\shallowXg$, this is because $\x$ is not restricted to be in $\shallowXg$ (it might be in $X \setminus \shallowXg$).

\paragraph{} \Cref{apx:hierarchy-solvers} provides an extensive analysis of the hierarchy between the solvers.



\newcommand{\xApprox}{\x^A}
\subsection{Preparations for the Proof}
In this section, we present another definition of the leximin approximation and prove that it is equivalent to the definition provided in the main paper. We will use this new definition in the proof of \Cref{lemma:alg1-shallow-solver}.
Specifically, we prove that an \dist $\xApprox \in X$ is an $\multApprox$-leximin-approximation if and only if $\mathbf{E}(\xApprox) \weaklyPreferred \mathbf{E}(\x)$ for any $\x \in \shallowXg$.

Recall that $s_d$ is the dummy-selection that gives all agents utility $0$, and that $x_d$ is the probability the \dist $\x$ assigns to the dummy-selection; and also that $\shallowXg$ is a subset of $X$ where $\x \in \shallowXg$ if $x_d \geq (1-\multApprox)$.
We start by defining two operations.

\begin{definition}[$\multApprox$-Upgrade]
    For $\x \in \shallowXg$, an \emph{$\multApprox$-upgrade} of $\x$ is the output vector of the following function:
    \begin{align*}
       \up{\x} = &\left( \up{\x,1},\ldots, \up{\x,|S|} \right) \\
       s.t. \Hquad & \up{\x,j} = 
        \begin{cases}
            \frac{1}{\multApprox} \cdot x_j &  \forall j \neq d\\
            1- \frac{1}{\multApprox} \sum_{j\neq d} x_j & \text{otherwise}
        \end{cases}
    \end{align*}
\end{definition}

\begin{lemma}\label{lemma:upgrade}
    Let $\x \in \shallowXg$, and let $\x^{\upS} := \up{\x}$ be its $\multApprox$-upgrade. ~~Then, $\x^{\upS} \in X$ and $\mathbf{E}(\x^{\upS}) = \frac{1}{\multApprox}\mathbf{E}(\x)$.
\end{lemma}

\begin{proof}
    As $\x \in \shallowXg$, we know that $x_j \geq 0$ for all $j$ and that $\sum_{j \neq d} x_j \leq \multApprox$. 
    
    For $j \neq d$, it is clear that $x^{\upS}_j \geq 0$.
    It is also true that $x^{\upS}_d \geq 0$:
    \begin{align*}
        x^{\upS}_d = 1- \frac{1}{\multApprox} \sum_{j\neq d} x_j \geq 1 - \frac{1}{\multApprox} \cdot \multApprox = 0
    \end{align*}
    We can also conclude that $\sum_{j} x^{\upS}_j = 1$:
    \begin{align*}
        \sum_{j = 1}^{|S|} x^{\upS}_j &= \sum_{j \neq d} x^{\upS}_j + x^{\upS}_d\\
        &= \sum_{j \neq d} \frac{1}{\multApprox} x_j + \left(1 - \frac{1}{\multApprox} \sum_{j\neq d} x_j\right) = 1
    \end{align*}
    Thus, $\x^{\upS} \in X$.
    
    In addition, regarding the expected utilities, $\mathbf{E}(\x^{\upS}) = \frac{1}{\multApprox}\mathbf{E}(\x)$ as the following holds for any $i \in N$: 
    \begin{align*}
       E_i(\x^{\upS}) &= \sum_{j = 1}^{|S|} x^{\upS}_j \cdot u_i(s_j) = \sum_{j = 1}^{|S|} \frac{1}{\multApprox} x_j \cdot u_i(s_j)\\
       &=  \frac{1}{\multApprox} \sum_{j = 1}^{|S|} x_j \cdot u_i(s_j) = \frac{1}{\multApprox}E_i(\x)
    \end{align*}
\end{proof}


\begin{definition}[$\multApprox$-Downgrade]
    For $\x \in X$, an \emph{$\multApprox$-downgrade} of $\x$ is the output vector of the following function:
    \begin{align*}
       \down{\x} = &\left( \down{\x,1},\ldots, \down{\x,|S|} \right) \\
       s.t. \quad & \down{\x,j} = 
        \begin{cases}
            \multApprox \cdot x_j &  \forall j \neq d\\
            1- \multApprox\sum_{j\neq d} x_j & \text{otherwise}
        \end{cases}
    \end{align*}
\end{definition}

\begin{lemma}\label{lemma:downgrade}
    Let $\x \in X$, and let $\x^{\downS} := \down{\x}$ be its $\multApprox$-downgrade. ~~Then, $\x^{\downS} \in \shallowXg$ and $\mathbf{E}(\x^{\downS}) = \multApprox\mathbf{E}(\x)$.
\end{lemma}

\begin{proof}
    As $\x \in X$, we know that $x_j \geq 0$ for all $j$ and that $\sum_{j \neq d} x_j \leq 1$. 
    
    For $j \neq d$, it is clear that $x^{\downS}_j \geq 0$.
    Also, $x^{\downS}_d \geq (1- \multApprox)$, since:
    \begin{align*}
        x^{\downS}_d &= 1- \multApprox \sum_{j\neq d} x_j \\
        &\geq 1 - \multApprox && \text{(as  $\sum_{j\neq d} x_j \leq 1 $)}
    \end{align*}
    We can also conclude that $\sum_{j} x^{\downS}_j = 1$:
    \begin{align*}
        \sum_{j = 1}^{|S|} x^{\downS}_j &= \sum_{j \neq d} x^{\downS}_j + x^{\downS}_d\\
        &= \sum_{j \neq d} \multApprox \cdot x_j + \left(1 - \multApprox \sum_{j\neq d} x_j\right) = 1
    \end{align*}
    Thus, $\x^{\downS} \in \shallowXg$. 
    
    In addition, $\mathbf{E}(\x^{\downS}) = \multApprox \mathbf{E}(\x)$, as the following holds for any $i \in N$: 
    \begin{align*}
       E_i(\x^{\downS}) &= \sum_{j = 1}^{|S|} x^{\downS}_j \cdot u_i(s_j) = \sum_{j = 1}^{|S|} \multApprox \cdot x_j \cdot u_i(s_j)\\
       & =  \multApprox \sum_{j = 1}^{|S|} x_j \cdot u_i(s_j) = \multApprox E_i(\x)
    \end{align*}
\end{proof}

\paragraph{}
Essentially, $\multApprox$-upgrade multiplies the probability of all non-dummy selections by $1/\multApprox$, and $\multApprox$-downgrade multiplies them by $\multApprox$. The operations are converse:

\begin{observation}
\label{obs:upgrade-downgrade}
    For all $\x\in \shallowXg$, the $\multApprox$-downgrade of the 
    $\multApprox$-upgrade of $\x$ equals $\x$.
    ~ Similarly, 
        for all $\x \in X$, The $\multApprox$-upgrade of the 
    $\multApprox$-downgrade of $\x$ equals $\x$.
    \begin{align*}
        &\forall \x \in \shallowXg \colon \Hquad \down{\up{\x}} = \x\\
        &\forall \x \in X \colon \quad \up{\down{\x}} = \x
    \end{align*}
\end{observation}

Upgrades and downgrades preserve the leximin order:
\begin{observation}
\label{obs:upgrade-leximin}
For two vectors $\x,\x' \in \shallowXg$, 
$\mathbf{E}(\x) \weaklyPreferred \mathbf{E}(\x')$
if-and-only-if the same relation holds for their $\multApprox$-upgrades $\mathbf{E}(\up{\x}) \weaklyPreferred \mathbf{E}(\up{\x'})$.

Similarly, for two vectors $\x,\x' \in X$,
$\mathbf{E}(\x) \weaklyPreferred \mathbf{E}(\x')$
if-and-only-if the same relation holds for their $\multApprox$-downgrades
 $\mathbf{E}(\down{\x}) \weaklyPreferred \mathbf    {E}(\down{\x'})$.
\end{observation}

\subsubsection{Relation Between $X$ and $\shallowXg$.}

Recall that an \dist, $\x^* \in X$, is leximin optimal\footnote{Notice that there might be different solutions that are leximin-optimal, but all of their expected vector are leximin-equivalent.} if $\mathbf{E}(\x^*) \weaklyPreferred \mathbf{E}(\x)$ for all $\x \in X$.

\begin{lemma}\label{lemma:threshold-is-opt-in-shallow}
    Let $\x^*$ be a leximin optimal \dist. Then, 
    $\mathbf{E}(\down{\x^*}) \weaklyPreferred \mathbf{E}(\x)$ for all $\x \in \shallowXg$.
\end{lemma}

\begin{proof}
Let $\x \in \shallowXg$. By \Cref{lemma:upgrade}, $\up{\x} \in X$. As $\x^*$ is leximin optimal, we get that
$\mathbf{E}(\x^*) \weaklyPreferred \mathbf{E}(\up{\x})$.
By \Cref{obs:upgrade-leximin},
$\mathbf{E}(\down{\x^*}) \weaklyPreferred \mathbf{E}(\down{\up{\x}}$; and by \Cref{obs:upgrade-downgrade}, $\mathbf{E}(\down{\up{\x}} = \mathbf{E}(\x)$. Thus, $\mathbf{E}(\down{\x^*}) \weaklyPreferred \mathbf{E}(\x)$.

\end{proof}





Now, recall that an \dist, $\xApprox \in X$, is an $\multApprox$-leximin-approximation if $\mathbf{E}(\xApprox) \weaklyPreferred \multApprox \cdot \mathbf{E}(\x)$ for all $\x \in X$.

\begin{lemma}\label{lemma:aprox-eqv-def}
    An \dist $\xApprox$ is an $\multApprox$-leximin-approximation if and only if $\mathbf{E}(\xApprox) \weaklyPreferred \mathbf{E}(\x)$ for all $\x \in \shallowXg$.
\end{lemma}

\begin{proof}
    Let $\x^*$ be a leximin optimal \dist. 

    Let $\xApprox$ be an $\multApprox$-leximin-approximation.
    By definition, $\mathbf{E}(\xApprox) \weaklyPreferred \multApprox \cdot \mathbf{E}(\x)$ for all $\x \in X$.
    Since $\x^* \in X$, we get that $\mathbf{E}(\xApprox) \weaklyPreferred \multApprox \cdot \mathbf{E}(\x^*)$.
    Now, consider $\down{x^*}$. By \Cref{lemma:downgrade}, $\mathbf{E}(\down{x^*}) = \multApprox \mathbf{E}(\x^*)$.
    This implies that $\mathbf{E}(\xApprox) \weaklyPreferred \mathbf{E}(\down{x^*})$.
    Together with \Cref{lemma:threshold-is-opt-in-shallow}, and by transitivity, this means that $\mathbf{E}(\xApprox) \weaklyPreferred \mathbf{E}(\x)$ for all $\x \in \shallowXg$.

    On the other hand, let $\xApprox$ be an \dist such that $\mathbf{E}(\xApprox) \weaklyPreferred \mathbf{E}(\x)$ for all $\x \in \shallowXg$. 
    Let $\xt \in X$. By \Cref{lemma:downgrade}, $\down{\xt} \in \shallowXg$ and $\mathbf{E}(\down{\xt}) = \multApprox \mathbf{E}(\xt)$.
    As $\down{\xt} \in \shallowXg$, we get that $\mathbf{E}(\xApprox) \weaklyPreferred \mathbf{E}(\down{\xt})$; and as $\mathbf{E}(\down{\xt}) = \multApprox \mathbf{E}(\xt)$. this implies that $\mathbf{E}(\xApprox) \weaklyPreferred \multApprox \mathbf{E}(\xt)$.
\end{proof}



\subsection{Proof of \Cref{lemma:alg1-shallow-solver}}
We can now use the new definition to prove a slightly different version of our  \Cref{lemma:alg1-shallow-solver}.

\begin{lemma}\label{lemma:main-general}
    Given an $\multApprox$-shallow-solver for \ref{eq:compact-OP}, Algorithm \ref{alg:basic-ordered-Outcomes} returns an \dist $\retSol$ such that $\mathbf{E}(\retSol) \weaklyPreferred \mathbf{E}(\x)$ for all $\x \in \shallowXg$.
\end{lemma}

Clearly, together with \Cref{lemma:aprox-eqv-def}, this proves our \Cref{lemma:alg1-shallow-solver} --- as required.

\begin{proof}
As a first observation, we note that $\x^n$ returned by the algorithm is a solution for \ref{eq:general-comp} that was solved in the last iteration. However, as constraints are only added along the way, it implies that:
\begin{observation}\label{obs:retsol-is-feas-to-all}
    $\retSol$ is a solution for the program \ref{eq:general-comp} that was solved in each one of the iterations $t=1,\ldots,n$. 
\end{observation}

Now, suppose by contradiction that there exists a $\xt \in \shallowXg$ such that $\mathbf{E}(\xt) \succ \mathbf{E}(\retSol)$.
By definition, there exists an integer $1 \leq k\leq n $ such that $\expectedValBy{i}{\xt} = \expectedValBy{i}{\retSol}$ for $i \leq k$, and $\expectedValBy{k}{\xt} > \expectedValBy{k}{\retSol}$.

As $\retSol$ is a solution for the program \ref{eq:general-comp} that was solved in the last iteration ($t=n$), we can conclude that $\sum_{i=1}^k \expectedValBy{i}{\retSol} \geq \sum_{i=1}^{k} z_i$ (by constraint (\ref{eq:general-comp}.3) if $k<n$ and by its objective otherwise).
This implies that the objective value of $\retSol$ for the program \ref{eq:general-comp} that was solved in $k$-th iteration at least $z_t$:
\begin{align}\label{eq:fk-to-zk}
    \sum_{i=1}^k \expectedValBy{i}{\retSol} - \sum_{i=1}^{k-1} z_i \geq z_k
\end{align}

We shall now see that $\xt$ is also a solution to this problem. Constraints (\ref{eq:general-comp}.1--2) are satisfied since $\xt \in \shallowXg \subseteq X$.
For Constraint (\ref{eq:general-comp}.3), we notice that the $(k-1)$ least expected values of $\xt$ equals to those of $\retSol$, which means that, for any $\ell < k$:
\begin{align*}
    \sum_{i=1}^{\ell} \expectedValBy{i}{\xt}  = \sum_{i=1}^{\ell} \expectedValBy{i}{\retSol} \geq \sum_{i=1}^{\ell} z_i
\end{align*}
Therefore, $\xt$ is also a solution for this program, and its objective value for it is:
\begin{align*}
    \sum_{i=1}^k \expectedValBy{i}{\xt} - \sum_{i=1}^{k-1} z_i
\end{align*}
We shall now see that this means that the objective value of $\xt$ is \emph{strictly-higher} than the objective value of the solution returned by the solver in this iteration, namely $z_k$.
\begin{align*}
    &\sum_{i=1}^k \expectedValBy{i}{\xt} - \sum_{i=1}^{k-1} z_i = \sum_{i=1}^{k-1} \expectedValBy{i}{\xt} + \expectedValBy{k}{\xt}- \sum_{i=1}^{k-1} z_i \\
    &\equWithExp{\text{By definition of $\xt$ for $i< k$}}{= \sum_{i=1}^{k-1} \expectedValBy{i}{\retSol} + \expectedValBy{k}{\xt}- \sum_{i=1}^{k-1} z_i}\\
    &\equWithExp{\text{By definition of $\xt$ for $k$}}{> \sum_{i=1}^{k-1} \expectedValBy{i}{\retSol} + \expectedValBy{k}{\retSol}- \sum_{i=1}^{k-1} z_i}\\
    &\equWithExp{\text{By \Cref{eq:fk-to-zk}}}{\geq z_t}
\end{align*}

But this contradicts the guarantees of our \emph{shallow} solver --- by definition, the objective value of the solution returned by the solver, namely $z_k$, is at least as high as the objective of any solution in $\feasC \cap \shallowXg$. 

\end{proof}


\subsection{A more general theorem}
The proof of \Cref{lemma:alg1-shallow-solver} does not use the specific structure of $X$, $\shallowXg$, or the functions $E_i$.
Therefore, we have in fact proved a more general theorem.

Let $X$ be any subset of $\mathbb{R}^m$ for some $m \in \mathbb{N}$,
let $Y$ be any subset of $X$,
and let $E_i$ for $i \in N$ be any functions from $X$ to $\mathbb{R}_{\geq 0}$.

Define a \emph{$Y$-shallow-solver for \ref{eq:compact-OP}}
as a solver that returns a solution $\x^t \in \feasC$ such that $\objCx{\x^t}\geq \objCx{\x}$ for all solutions $\x \in \feasC \cap Y$. Then:
\begin{lemma}
    Given a $Y$-shallow-solver for \ref{eq:compact-OP}, Algorithm \ref{alg:basic-ordered-Outcomes} returns an $\x^n \in X$ such that $\mathbf{E}(\x^n) \weaklyPreferred \mathbf{E}(\x)$ for all $\x \in Y$.
\end{lemma}
The proof is identical to that of \Cref{lemma:alg1-shallow-solver}. 








\subsection{Hierarchy of the Solvers for \ref{eq:compact-OP}}\label{apx:hierarchy-solvers}
Consider the two non-exact types: approximately-optimal, and shallow.
First, both are relaxations of the exact solver. 
One can easily verify it by taking $\multApprox=1$ for the approximately-optimal solver and by taking $\shallowXg = X$ for the shallow solver.

In addition, we shall now prove that every $\multApprox$-approximately-optimal solver
is an $\multApprox$-shallow-solver.

\begin{claim}\label{lemma:approx-to-shallow}
    Any solution $\x^t$ that satisfies the requirements of the approximately-optimal solver also satisfies the requirements of the shallow solver.
\end{claim}

\begin{proof}
    Let $\x^t$ be a solution that satisfies the requirements of the approximately-optimal solver.
    By definition, $\x^t \in \feasC$ and $\objCx{\x^t} \geq \multApprox \cdot\objCx{\x}$ for all $\x \in \feasC$.
    We need to prove that $\objCx{\x^t} \geq \objCx{\x}$ for all $\x \in \feasC \cap \shallowXg$.

    Suppose by contradiction that there exists $\x \in \feasC \cap \shallowXg$ such that $\objCx{\x} > \objCx{\x^t}$.
    Consider $\up{\x}$. By \Cref{lemma:upgrade}, $\up{\x} \in X$ and $\mathbf{E}(\up{\x}) = \frac{1}{\multApprox}\mathbf{E}(\x)$.
    We shall now see that $\up{\x} \in \feasC$. 
    Constraints (\ref{eq:general-comp}.1--2) are satisfied as $\up{\x} \in X$.
   Constraint (\ref{eq:general-comp}.3) is satisfied as it does by $\x$ and as all the expected values of $\up{\x}$ are at least as those of $\x$; specifically, for any $\ell < t$:
    \begin{align*}
        \sum_{i=1}^{\ell} \expectedValBy{i}{\up{\x}}  = \frac{1}{\multApprox} \sum_{i=1}^{\ell} \expectedValBy{i}{\x} \geq \sum_{i=1}^{\ell} \expectedValBy{i}{\x} \geq \sum_{i=1}^{\ell} z_i
    \end{align*}
    However, $\objCx{\up{\x}} > \frac{1}{\multApprox}\objCx{\x^t}$:
    \begin{align*}
        &\objCx{\up{\x}} = \sum_{i=1}^t \expectedValBy{i}{\up{\x}} - \sum_{i=1}^{t-1} z_i\\
        &= \frac{1}{\multApprox} \expectedValBy{i}{\x} - \sum_{i=1}^{t-1} z_i \geq \frac{1}{\multApprox} \left(\expectedValBy{i}{\x} - \sum_{i=1}^{t-1} z_i\right)\\
        &= \frac{1}{\multApprox} \objCx{\x}> \frac{1}{\multApprox} \objCx{\x^t}
    \end{align*}
    Which means that $\multApprox \cdot\objCx{\up{\x}} > \objCx{\x^t}$ --- in contradiction to $\x^t$ being approximately-optimal.

    \end{proof}

Next, we prove that the opposite is true for $t=1$, but only for $t=1$.
\begin{claim}\label{lemma:t=1-shallow-to-approx}
    For $t=1$, any solution $\x^t$ that satisfies the requirements of the shallow solver also satisfies the requirements of the approximately-optimal solver.
\end{claim}

\begin{proof}
    Let $\x^t$ be a solution that satisfies the requirements of the shallow solver. 
     By definition, $\x^t \in \feasC$ and $\objCx{\x^t} \geq \objCx{\x}$ for all $\x \in \feasC \cap \shallowXg$.
    We need to prove that $\objCx{\x^t} \geq \multApprox \cdot\objCx{\x}$ for all $\x \in \feasC$.
    
    Suppose by contradiction that there exists $\x \in \feasC$ such that $\multApprox \cdot\objCx{\x} > \objCx{\x^t}$.
    Consider $\down{\x}$. By \Cref{lemma:downgrade}, $\down{\x} \in \shallowXg$ and $\mathbf{E}(\down{\x}) = \multApprox \mathbf{E}(\x)$.
    We shall now see that $\down{\x} \in \feasC$. 
    Constraints (\ref{eq:general-comp}.1--2) are satisfied as $\down{\x} \in \shallowXg \subseteq X$.
   Constraint (\ref{eq:general-comp}.3) is empty for $t=1$ and is therefore vacuously satisfied by $\down{\x}$.
   However, as $t=1$, we get that $\objCx{\down{\x}} > \objCx{\x^t}$:
    \begin{align*}
        &\objCx{\down{\x}} = \expectedValBy{1}{\down{\x}}\\
        &= \multApprox\expectedValBy{1}{\x}= \multApprox \cdot \objCx{\x}>  \objCx{\x^t}
    \end{align*}
    This in contradiction to $\x^t$ being the returned solution by the shallow solver.
    
\end{proof}

Lastly, we prove that for $t=1$ a solution returned by the shallow solver might not satisfies the requirements of the approximately-optimal solver.
This implies that the requirements for our algorithm are weaker than those of \citet{hartman2023leximin}.

\begin{claim}
    There exists an instance for which the solution returned by the shallow solver does not satisfy the requirements of the approximately-optimal solver.
\end{claim}

\begin{proof}
    Let $\multApprox = 0.9$, $N = \{1,2\}$, $S = \{s_1,s_2,s_d\}$ and $u_1,u_2$ as follows:
    \begin{align*}
        &u_1(s_1) = 10 && u_1(s_2) = 0  &&& u_1(s_d) = 0 \\
        &u_2(s_1) = 10 && u_1(s_2) = 1000  &&& u_1(s_d) = 0 
    \end{align*}

    At the first iteration, $t=1$, the optimal solution that maximizes the minimum expected value is $(1,0,0)$ with objective value $10$. 
    However, the optimal solution \emph{among $\shallowXg$} is $(0.9,0,0.1)$ with objective value $9$.
    Thus, any solution with a minimum expected value $9$ satisfies the requirements of the shallow solver (and also the requirements of the approximately-optimal solver by \cref{lemma:t=1-shallow-to-approx}).
    %
    %
    %
    %
    
    Suppose that the solver returned this solution, $(0.9,0,0.1)$, and so $z_1 := 9$.
    Then, at the second iteration $t=2$, Constraint (\ref{eq:general-comp}.3) says that the smallest expected value is at least $9$.
    As $s_1$ is the only selection that gives agent $1$ a positive utility of $10$, any solution that satisfies this constraint must give this selection a probability of at least $0.9$.
    Thus, the only solution in $\shallowXg$ that satisfies this constraint is $(0.9,0,0.1)$ with objective value $9$:
    \begin{align*}
        &\objCx{(0.9,0,0.1)} \\
        &= \expectedValBy{1}{(0.9,0,0.1)} + \expectedValBy{2}{(0.9,0,0.1)} - z_1 \\
        &= 9 + 9 - 9 = 9
    \end{align*}
    Therefore, any solution with objective value $9$ satisfies the requirements of the shallow solver.
    Specifically, $(0.9,0,0.1)$ does.
    
    However, as we have the solution $(0.9,0.1,0)$ with objective value $109$, the solution $(0.9,0,0.1)$ does \emph{not} satisfy the requirements of the approximately-optimal solver ---  as $0.9 \times 109 > 9$.
    
\end{proof}

%% file: appendices/feas-oracle.tex
\section{Designing a Shallow Solver for \ref{eq:compact-OP} ~~(Proof of \Cref{lemma:red-shallow-solver-to-feas-test})}\label{apx:desiging-shallow}
Recall \Cref{lemma:red-shallow-solver-to-feas-test}:
\begin{lemma*}
Given an $\multApprox$-approximate-feasibility-oracle for \ref{eq:compact-OP} (\rom{3}), an $\multApprox$-approximate black-box for the utilitarian welfare (\rom{1}),
and an arbitrary vector in $\feasC$.~~ An efficient $\multApprox$-shallow-solver for \ref{eq:compact-OP} (\rom{2}) can be design.
\end{lemma*}

\renewcommand\algorithmiccomment[1]{%
  \hfill $\rhd$ {#1}%
}

\begin{algorithm}[tbp]
\caption{$\multApprox$-Shallow Solver for \eqref{eq:compact-OP}}
\label{alg:shallow-solver}
\textbf{Input}: 
 An integer $t \in N$ and rationals $z_1,\ldots, z_{t-1}$.
\\
 * If $t\geq 2$, then also $\x^{t-1}$.
 \\
\textbf{Oracles}: 
 an $\multApprox$-approximate-feasibility-oracle for \ref{eq:compact-OP} (\rom{3}), an $\multApprox$-approximate black-box for the utilitarian welfare (\rom{1}), and an arbitrary vector in $\feasC$.\\
\textbf{Parameter}: An error factor $\epsilon > 0$.\\
\begin{algorithmic}[1] 

\STATE Let 
retSol := the given arbitrary vector in $\feasC$.\label{alg-line:bounds-s}
\STATE Let $l :=\objCx{\text{retSol}}$
\STATE Let $U'$ be the utilitarian welfare obtained by using the approximate black-box with $c_i =1$ for $i \in N$.

\STATE Let $u := \frac{1}{\multApprox} U'$.\\ \quad \label{alg-line:bounds-e}

\WHILE{$u - l > \epsilon$} \label{alg-line:search-s}

\STATE Let $\ztCons :=(l+u) /2$.
\STATE Let (ans, $\xt$) be the answer of the  approximate-feasibility-oracle for the value $\ztCons$.
\IF{ans = Feasible}
\STATE update $l := \ztCons$.
\STATE update retSol := $\xt$.
\ELSE[ans = Infeasible-Under-$\shallowXg$]
\STATE update $u := \ztCons$.
\ENDIF
\ENDWHILE \\ \quad \label{alg-line:search-e}
\RETURN retSol \label{alg-line:ret}
\end{algorithmic}
\end{algorithm}

\begin{proof}
    The solver is described in \Cref{alg:shallow-solver}.  It performs a binary search over the potential objective-values $\ztCons$ for the program \ref{eq:compact-OP}. 
    
We start by proving that performing a binary search makes sense as we have monotonicity. 
First, 
if some objective value $\ztCons$ is Feasible (i.e., there exits a solution $\x \in \feasC$ with objective value at least $z_t$), 
then any value $\ztCons^- \leq \ztCons$ is also Feasible.  To see that, assume that $\ztCons$ is Feasible, and let $\x \in \feasC$ be a solution with objective value at least $\ztCons$. Clearly, the same $\x$ also has an objective value at least $\ztCons^-$.
    Similarly, if $\ztCons$ is Infeasible, then any value $\ztCons^+ \geq \ztCons$ is also Infeasible.

    Lines \ref{alg-line:bounds-s}--\ref{alg-line:bounds-e} set bounds for the binary search.
    
    As a lower bound, we use the objective value of the  solution given as input.
    
    For an upper bound, we use the given 
    $\multApprox$-approximate black-box for the utilitarian welfare with $c_i = 1$ for all $i \in N$ (i.e., with the original utilities $u_i$), to obtain a value $U'$.
    By definition of $\multApprox$-approximation, $\frac{1}{\multApprox} U'$ is an upper bound on the sum of utilities. 
    Recall that the objective function is the sum of the smallest $t$ utilities minus a positive constant. Thus, it is clear that an upper bound on the sum of all utilities can be used also as an upper bound on the maximum objective value.

    To perform the search, we use $u$ that holds the upper bound and $l$ that holds the lower bound; one of which is updated at each iteration:  
    \rmark{In addition, we use retSol that is initialized 
    to the  solution given as input.
    It is updated only in some of the iterations as follows.}

     At each iteration, we examine the midpoint value between the upper and lower bounds, $\ztCons = \frac{u + l}{2}$, and query the oracle about this value. If the value is determined to be Feasible, we update retSol to the solution returned by the oracle, and the lower bound $l := \ztCons$ to search for larger values. Otherwise, we update the upper bound $u:= \ztCons$ to search for smaller values. 
     \rmark{We stop the search when $l$ and $u$ are sufficiently close --- for now let us assume that we stop it when $l = u$; we revisit this issue extensively in \Cref{apx:binary-search-error}.} 

    To prove that the
    solver acts as described above, we need to show that 
    (a) the returned $\x$ (retSol) is feasible for \ref{eq:compact-OP}; and
    (b) $\objCx{\text{retSol}} \geq \objCx{\x}$ for all $\x \in \feasC \cap \shallowX$.


(a) Suppose first that the solver returns the initial value of retSol - by definition, it is feasible for \ref{eq:compact-OP}.
If the solver returns a modified value of retSol, then this value must have been returned by the approximate-feasibility-oracle. By definition of the oracle, the returned solution is feasible.


(b) We first note that $\objCx{\text{retSol}}$ is always at least the lower bound $l$.
Next, suppose by contradiction that there exists $\x \in \feasC \cap \shallowX $ such that  $\objCx{\x} > \objCx{\text{retSol}}$.
Therefore, $\objCx{\x} > l$.
As we stop the search when $l=u$, this implies $\objCx{\x} > u$.
But $u$ is a value for which the approximate feasibility oracle has asserted Infeasible-Under-$\shallowXg$.
By monotonicity, as $\objCx{\x} > u$, this implies that$\objCx{\x}$ is Infeasible-Under-$\shallowXg$ too. 
But $\x \in \feasC \cap \shallowX$ and has an objective value $\objCx{\x}$ --- a contradiction.

\end{proof}

\newcommand{\epsVec}{\Delta(\epsilon)}

\subsection{The Binary Search Error}\label{apx:binary-search-error}

The shallow solver described in \Cref{alg:shallow-solver} actually has an additional additive error $\epsilon >0$ that arises from the binary search ending condition.
So far, we have assumed that this error is negligible, as it can be made smaller than $\epsilon$ in $\mathcal{O} (\log \frac{1}{\epsilon})$, for any $\epsilon > 0$.
Here we provide a more accurate analysis, that does not neglect this error.
    
Formally, let $\epsilon > 0$ be the error obtained from the binary search;  that is, we stop the search when $u - l \leq \epsilon$.

We claim that, with this modification,
the solver described by \Cref{alg:shallow-solver} returns 
    a solution retSol such that $\objCx{\text{retSol}} \geq \objCx{\x} - \epsilon$ for all solutions $\x \in \feasC \cap \shallowX$. That is, the solver returns a solution whose objective is at least the maximum objective among the subset $\shallowX$ \emph{minus $\epsilon$}.
    Clearly, this reduces to the definition of $\multApprox$-shallow solver when $\epsilon$ is negligible.
    We call it an $(\multApprox, \epsilon)$-shallow-solver.

    \begin{tcolorbox}[left=2pt, right=2pt,  
colback=black!5!white,colframe=black!50!black, colbacktitle=black!75!black,
title=\textbf{(\rom{2}')~$(\multApprox, \epsilon)$-Shallow-Solver for \ref{eq:compact-OP}}]
  \textbf{Input:}~ An integer $t \in N$ and rationals $z_1,\ldots, z_{t-1}$.
  \tcblower
  \textbf{Output:} A solution
  $\x^t \in \feasC$ such that\\ $\objCx{\x^t} \geq \objCx{\x} - \epsilon$ for any $\x \in \feasC \cap \shallowX$.
\end{tcolorbox}


    We also claim that using this type of solver in \Cref{alg:basic-ordered-Outcomes} affects the guarantees on its output as follows. 
    Let $\mathbf{1}_n$ be a vector of size $n$, where each one of its component is $1$.
    Then, \Cref{alg:basic-ordered-Outcomes} returns a solution $\retSol$ such that $\mathbf{E}(\retSol) \weaklyPreferred \multApprox \cdot \mathbf{E}(\x) - \epsilon \cdot \mathbf{1}_n$ for all $\x \in X$.
    %
    
    Comparing the the previous guarantees, here we have an additional subtraction of $\epsilon$ from each component. 
    It is again clear, that this reduces to the definition of $\multApprox$-leximin-approximation when $\epsilon$ is negligible.
    We call it an $(\multApprox, \epsilon)$-leximin-approximation.
    
     \begin{definition}[An $(\multApprox, \epsilon)$-leximin-approximation]
        An \dist, $\xApprox \in X$, is an \rmark{$(\multApprox, \epsilon)$}-leximin-approximation if $\mathbf{E}(\xApprox) \weaklyPreferred \multApprox \cdot \mathbf{E}(\x)  - \epsilon \cdot \mathbf{1}_n$ for all $\x \in X$.
    \end{definition} 

    \paragraph{} We shall now provide a re-analysis of \Cref{alg:shallow-solver} and \Cref{alg:basic-ordered-Outcomes}, which considers the binary search error.

    \paragraph{Re-Analysis of
    \Cref{alg:shallow-solver}.} We prove a more accurate version of \Cref{lemma:red-shallow-solver-to-feas-test}:
    \begin{lemma}
    Given an $\multApprox$-approximate-feasibility-oracle for \ref{eq:compact-OP} (\rom{3}), an $\multApprox$-approximate black-box for the utilitarian welfare (\rom{1}),
    and an arbitrary vector in $\feasC$.~~ An efficient \rmark{$(\multApprox,\epsilon)$-shallow-solver} for \ref{eq:compact-OP} (\rom{2}) can be design.
    \end{lemma}

    \begin{proof}
        The proof is similar to the original proof except (b), which here becomes $\objCx{\text{retSol}} \geq \objCx{\x} - \epsilon$ for all $\x \in \feasC \cap \shallowX$.
        
        Suppose by contradiction that there exists $\x \in \feasC \cap \shallowX $ such that  $\objCx{\x} - \epsilon > \objCx{\text{retSol}}$.
        Then, as $\objCx{\text{retSol}} \geq l $, we get that $\objCx{\x} > l + \epsilon$.
        As we stop the search when $u-l \leq \epsilon$, this implies $\objCx{\x} > u$.
But $u$ is a value for which the approximate feasibility oracle has asserted Infeasible-Under-$\shallowXg$.
By monotonicity, as $\objCx{\x} > u$, this implies that$\objCx{\x}$ is Infeasible-Under-$\shallowXg$ too. 
But $\x \in \feasC \cap \shallowX$ and has an objective value $\objCx{\x}$ --- a contradiction.
\end{proof}

\paragraph{Re-Analysis of
    \Cref{alg:basic-ordered-Outcomes}.}
    We start by proving the following lemma that extends \Cref{lemma:aprox-eqv-def}, and provides another equivalent definition for $(\multApprox,\epsilon)$-leximin-approximation:
    
    \begin{lemma}\label{lemma:approx-def-with-add}
    An \dist $\xApprox$ is an \rmark{$(\multApprox,\epsilon)$}-leximin-approximation if and only if $\mathbf{E}(\xApprox) \weaklyPreferred \mathbf{E}(\x) - \rmark{\epsilon \cdot \mathbf{1}_n}$ for all $\x \in \shallowXg$.
\end{lemma}

\begin{proof}
    Let $\x^*$ be a leximin optimal \dist. 

    Let $\xApprox$ be an \rmark{$(\multApprox,\epsilon)$}-leximin-approximation.
    By definition, $\mathbf{E}(\xApprox) \weaklyPreferred \multApprox \cdot \mathbf{E}(\x) - \rmark{\epsilon \cdot \mathbf{1}_n}$ for all $\x \in X$.
    Since $\x^* \in X$, we get that $\mathbf{E}(\xApprox) \weaklyPreferred \multApprox \cdot \mathbf{E}(\x^*) - \rmark{\epsilon \cdot \mathbf{1}_n}$.
    Now, consider $\down{x^*}$. By \Cref{lemma:downgrade}, $\mathbf{E}(\down{x^*}) = \multApprox \mathbf{E}(\x^*)$.
    This implies that $\mathbf{E}(\xApprox) \weaklyPreferred \mathbf{E}(\down{x^*}) - \rmark{\epsilon \cdot \mathbf{1}_n}$.
    By \Cref{lemma:threshold-is-opt-in-shallow}, $\mathbf{E}(\down{\x^*}) \weaklyPreferred \mathbf{E}(\x)$ for all $\x \in \shallowXg$.
    \rmark{Since subtracting the same constant from each component preserves the leximin order, it follows that $\mathbf{E}(\down{\x^*})  - \rmark{\epsilon \cdot \mathbf{1}_n} \weaklyPreferred \mathbf{E}(\x)  - \rmark{\epsilon \cdot \mathbf{1}_n} $ for all $\x \in \shallowXg$.
    By transitivity, this means that $\mathbf{E}(\xApprox) \weaklyPreferred \mathbf{E}(\x)  - \epsilon \cdot \mathbf{1}_n$ for all $\x \in \shallowXg$.}

    On the other hand, let $\xApprox$ be an \dist such that $\mathbf{E}(\xApprox) \weaklyPreferred \mathbf{E}(\x) - \rmark{\epsilon \cdot \mathbf{1}_n}$ for all $\x \in \shallowXg$. 
    Let $\xt \in X$. By \Cref{lemma:downgrade}, $\down{\xt} \in \shallowXg$ and $\mathbf{E}(\down{\xt}) = \multApprox \mathbf{E}(\xt)$.
    As $\down{\xt} \in \shallowXg$, we get that $\mathbf{E}(\xApprox) \weaklyPreferred \mathbf{E}(\down{\xt})  - \rmark{\epsilon \cdot \mathbf{1}_n} $; and as $\mathbf{E}(\down{\xt}) = \multApprox \mathbf{E}(\xt)$. this implies that $\mathbf{E}(\xApprox) \weaklyPreferred \multApprox \mathbf{E}(\xt) - \rmark{\epsilon \cdot \mathbf{1}_n} $.
\end{proof}

We shall now prove the following lemma that extends \Cref{lemma:main-general}:

\begin{lemma}
    Given an \rmark{$(\multApprox, \epsilon)$}-shallow-solver for \ref{eq:compact-OP}, \Cref{alg:basic-ordered-Outcomes} returns an \dist $\retSol$ such that $\mathbf{E}(\retSol) \weaklyPreferred \mathbf{E}(\x)  - \rmark{\epsilon \cdot \mathbf{1}_n} $ for all $\x \in \shallowXg$.
\end{lemma}

Together with \Cref{lemma:approx-def-with-add}, this proves that \Cref{alg:basic-ordered-Outcomes} returns an $(\multApprox, \epsilon)$-leximin-approximation.

\begin{proof}
    suppose by contradiction that there exists a $\xt \in \shallowXg$ such that $\mathbf{E}(\xt)   - \rmark{\epsilon \cdot \mathbf{1}_n} \succ \mathbf{E}(\retSol)$.
By definition, there exists an integer $1 \leq k\leq n $ such that $\expectedValBy{i}{\xt} \rmark{-\epsilon} = \expectedValBy{i}{\retSol}$ for $i \leq k$, and $\expectedValBy{k}{\xt} \rmark{-\epsilon} > \expectedValBy{k}{\retSol}$.

As $\retSol$ is a solution for the program \ref{eq:general-comp} that was solved in the last iteration ($t=n$), we can conclude that $\sum_{i=1}^k \expectedValBy{i}{\retSol} \geq \sum_{i=1}^{k} z_i$ (by constraint (\ref{eq:general-comp}.3) if $k<n$ and by its objective otherwise).
This implies that the objective value of $\retSol$ for the program \ref{eq:general-comp} that was solved in $k$-th iteration at least $z_t$:
\begin{align}\label{eq:fk-to-zk-binary-error}
    \sum_{i=1}^k \expectedValBy{i}{\retSol} - \sum_{i=1}^{k-1} z_i \geq z_k
\end{align}

We shall now see that $\xt$ is also a solution to this problem. Constraints (\ref{eq:general-comp}.1--2) are satisfied since $\xt \in \shallowXg \subseteq X$.
For Constraint (\ref{eq:general-comp}.3), we notice that the $(k-1)$ least expected values of $\xt$ \rmark{are higher than} those of $\retSol$, which means that, for any $\ell < k$:
\begin{align*}
    \sum_{i=1}^{\ell} \expectedValBy{i}{\xt} =  \sum_{i=1}^{\ell} \left(\expectedValBy{i}{\retSol} + \epsilon \right) \geq \sum_{i=1}^{\ell} \expectedValBy{i}{\retSol}\geq \sum_{i=1}^{\ell} z_i
\end{align*}
Therefore, $\xt$ is also a solution for this program, and its objective value for it is:
\begin{align*}
    \sum_{i=1}^k \expectedValBy{i}{\xt} - \sum_{i=1}^{k-1} z_i
\end{align*}
We shall now see that this means that the objective value of $\xt$ \rmark{is higher by more than $\epsilon$} than the objective value of the solution returned by the solver in this iteration, namely $z_k$.
\begin{align*}
    &\sum_{i=1}^k \expectedValBy{i}{\xt} - \sum_{i=1}^{k-1} z_i = \sum_{i=1}^{k-1} \expectedValBy{i}{\xt} + \expectedValBy{k}{\xt}- \sum_{i=1}^{k-1} z_i \\
    &\equWithExp{\text{By definition of $\xt$ for $i< k$}}{= \sum_{i=1}^{k-1} \left(\expectedValBy{i}{\retSol} \rmark{+ \epsilon} \right)+ \expectedValBy{k}{\xt}- \sum_{i=1}^{k-1} z_i}\\
    & \geq \sum_{i=1}^{k-1} \expectedValBy{i}{\retSol} + \expectedValBy{k}{\xt}- \sum_{i=1}^{k-1} z_i\\
    &\equWithExp{\text{By definition of $\xt$ for $k$}}{> \sum_{i=1}^{k-1} \expectedValBy{i}{\retSol} + \expectedValBy{k}{\retSol} \rmark{+ \epsilon}- \sum_{i=1}^{k-1} z_i}\\
    & = \sum_{i=1}^{k} \expectedValBy{i}{\retSol}  - \sum_{i=1}^{k-1} z_i \rmark{+ \epsilon}\\
    &\equWithExp{\text{By \Cref{eq:fk-to-zk-binary-error}}}{\geq z_t \rmark{+ \epsilon}}
\end{align*}

But this contradicts the guarantees of our \emph{shallow} solver --- by definition, the objective value of the solution returned by the solver, namely $z_k$, is at least as high as the objective of any solution in $\feasC \cap \shallowXg$ \rmark{minus $\epsilon$}. 
\end{proof}

%% file: appendices/constraints-eqv.tex
\section{Equivalence Between \eqref{eq:min-sum-OP} and \eqref{eq:app-vsums-OP} ~~ (Including Proof of \Cref{lemma:red-approx-solver-for-P3-to-solver-for-P2})}\label{apx:eqv-linear}

The proof uses the following lemma, which considers general vectors:
\begin{lemma}\label{lemma:constaints-eqv}
    Let $c \in \mathbb{R}_{\geq 0}$ be a non-negative constant,  $\vv \in \mathbb{R}^{N}$ any vector, and $k \in N$. 
    Then, 
    \begin{align}
        \sum_{i=1}^k \orderedVby{i} \geq c \label{eq:eqv-first-side} 
    \end{align}
    if and only if there exist $y_k \in \mathbb{R}$ and $\mathbf{m}_k \in \mathbb{R}^N$ s.t.
    \begin{align}
        &k y_k - \sum_{i=1}^n m_{k,i} \geq c\label{eq:aux-1}\\
        &m_{k,i} \geq y_k - v_i && \forall i \in N\label{eq:aux-2}\\
        &m_{k,i} \geq 0 && \forall i \in N\label{eq:aux-3}
    \end{align}
\end{lemma}
We note that this proof simplifies the proof of Lemma 7 in \cite{hartman2023leximin}.

\begin{proof}
    Assume that \Cref{eq:eqv-first-side} holds, that is: $\sum_{i=1}^k \orderedVby{i} \geq c$.
    Let 
    \begin{align*}
        \quad y_{k} &:= \orderedVby{k}
        \\
        m_{k,i} &:= \max(0, \orderedVby{k} -v_i) && \forall i \in N
    \end{align*}

    It is easy to see that \Cref{eq:aux-3} is satisfied.
    For \Cref{eq:aux-2}, observe that:
    \begin{align*}
        m_{k,i} &= \max(0, \orderedVby{k} -v_i) \\
        &\geq \orderedVby{k} -v_i= y_{k}-v_i
    \end{align*}
    Now, consider the sum $\sum_{i \in N} m_{k,i}$, we can change the order of the elements as follows:
    \begin{align*}
        &\sum_{i \in N} m_{k,i} = \sum_{i \in N} \max(0, \orderedVby{k} -v_i)\\
        &= \sum_{i \in N} \max(0, \orderedVby{k} -\orderedVby{i})
    \end{align*}
    
    As $\orderedVby{k}$ is the $k$-th least value of $\vv$, we get that $\orderedVby{k} - \orderedVby{i} \geq 0$ for any $i < k$, that $\orderedVby{k} - \orderedVby{k} = 0$; and that $\orderedVby{k} - \orderedVby{i} \leq 0$ for any $i > k$. It follows that:
    \begin{align*}
        &\sum_{i \in N} m_{k,i} = \sum_{i =1}^{k-1} (\orderedVby{k} -\orderedVby{i}) \\
        &= (k-1)\orderedVby{k} - \sum_{i=1}^{k-1} \orderedVby{i}
    \end{align*}
    We can now prove that \Cref{eq:aux-2} is satisfies as well:
    \begin{align*}
        &ky_k - \sum_{i \in N} m_{k,i} \\
        &= k\orderedVby{k} -(k-1)\orderedVby{k} + \sum_{i=1}^{k-1} \orderedVby{i} \\
        &= \sum_{i=1}^k \orderedVby{i} \geq c && \text{(As \Cref{eq:eqv-first-side} holds)}
    \end{align*}
---\\
    On the other hand, assume that there exist $y_k \in \mathbb{R}$ and $\mathbf{m}_k \in \mathbb{R}^N$ that satisfy Equations (\ref{eq:aux-1}--\ref{eq:aux-3}).

    By Equations (\ref{eq:aux-2}--\ref{eq:aux-3}), we get that $m_{k,i} \geq \max(0,y_k-v_i)$ for $i \in N$. Using the same technique of changing the order of the elements we can conclude that:
    \begin{align*}
        \sum_{i\in N}m_{k,i} \geq \sum_{i\in N} \max(0,y_k-v_i) = \sum_{i\in N} \max(0,y_k-\orderedVby{i}) 
    \end{align*}
    Now, consider the left hand side of \Cref{eq:aux-1}, it follows that:
    \begin{align*}
        &k y_k - \sum_{i \in N} m_{k,i} \leq k y_k - \sum_{i\in N} \max(0,y_k-\orderedVby{i}) \\
        & \leq k y_k - \sum_{i=1}^k \max(0,y_k-\orderedVby{i}) \\
        &= \sum_{i=1}^k\left(y_k - \max(0,y_k-\orderedVby{i})\right)
    \end{align*}
    However, each element in the sum can be simplified to $\min(y_k,\orderedVby{i})$: if $\max(0,y_k-\orderedVby{i}) = 0$ (which means that $y_k \leq \orderedVby{i}$) then this element gives $y_k - 0 = y_k$, and otherwise  (if $\max(0,y_k-\orderedVby{i}) > 0 $ which means that $y_k > \orderedVby{i}$) it gives $y_k - (y_k - \orderedVby{i}) = \orderedVby{i}$.
    Which means that we can conclude that:
    \begin{align*}
        & k y_k - \sum_{i \in N} m_{k,i} \leq \sum_{i=1}^k\left(y_k - \max(0,y_k-\orderedVby{i})\right)\\
        & = \sum_{i=1}^k\min(y_k, \orderedVby{i}) \leq \sum_{i=1}^k \orderedVby{i}
    \end{align*}
    However, by \Cref{eq:aux-1}, $k y_k - \sum_{i \in N} m_{k,i} \geq c$, so we can conclude that $\sum_{i=1}^k \orderedVby{i} \geq c$ --- as required. 
\end{proof}
    
Using this general lemma, we now prove an equivalence between Constraint (\ref{eq:min-sum-OP}.2) and the set of Constraints (\ref{eq:app-vsums-OP}.2--4):
    \begin{lemma}\label{lemma:linear-const-eqv}
        $\mathbf{x}$ satisfies Constraint (\ref{eq:min-sum-OP}.2) if-and-only-if there exist $y_{\ell}$ and $m_{\ell,i}$ for $1 \leq \ell \leq t$ and $1 \leq i \leq n$ such that $\left(\mathbf{x}, \mathbf{y}, \mathbf{m}\right)$ satisfies (\ref{eq:app-vsums-OP}.2--4).
    \end{lemma}

\begin{proof}

    Let $\x \in \mathbb{R}^{|S|}$ that satisfies Constraint (\ref{eq:min-sum-OP}.2)  --- that is,
    \begin{align}
        &\sum_{i=1}^{\ell} \expectedValBy{i}{\x} \geq \sum_{i=1}^{\ell} z_i && \forXinY{\ell}{t}\label{eq:con3}
    \end{align}
    Let $1\leq \ell \leq t$. Combining \Cref{lemma:constaints-eqv} with \Cref{eq:con3} where $c := \sum_{i=1}^{\ell} z_i$, $\vv := \mathbf{E}(x)$ and $k:=\ell$, we get this it possible \emph{if and only if} there exist $y_{\ell} \in \mathbb{R}$ and $\mathbf{m}_{\ell} \in \mathbb{R}^N_{\geq 0}$ s.t.
    \begin{align*}
        &{\ell} y_{\ell} - \sum_{i=1}^n m_{{\ell},i} \geq \sum_{i=1}^{\ell} z_i\\
        &m_{{\ell},i} \geq y_{\ell} - E_i(\x) && \forall i \in N\\
        &m_{k,i} \geq 0 && \forall i \in N
    \end{align*}
    
    Putting it all together, we get that 
    $\mathbf{x}$ satisfies Constraint (\ref{eq:min-sum-OP}.2) if-and-only-if there exist $y_{\ell}$ and $m_{\ell,i}$ for $1 \leq \ell \leq t$ and $1 \leq i \leq n$ such that $\left(\mathbf{x}, \mathbf{y}, \mathbf{m}\right)$ satisfies Constraints (\ref{eq:app-vsums-OP}.2--4) --- as required.
    %
    %
    %
    %
\end{proof}


It is simple to verify that this implies our \Cref{lemma:red-approx-solver-for-P3-to-solver-for-P2}
\begin{lemma*}
Let $\left(\xRet, \yRet, \mRet\right)$ 
be a poly-sparse  $\frac{1}{\multApprox}$-approximately-optimal solution for \ref{eq:app-vsums-OP}.~~
Then, $\xRet$ is a poly-sparse $\frac{1}{\multApprox}$-approximately-optimal solution for \ref{eq:min-sum-OP}.
\end{lemma*}
\begin{proof}
     \rmark{To prove that $\xRet$ is a poly-sparse $\frac{1}{\multApprox}$-approximately-optimal solution for \ref{eq:min-sum-OP}, we need to prove that (a) $\xRet$ is a poly-sparse solution for \ref{eq:min-sum-OP}, and (b) its objective value is $\frac{1}{\multApprox}$-approximately-optimal.}
     
     \rmark{(a) First, $\xRet$ is a poly-sparse vector since $\left(\xRet, \yRet, \mRet\right)$ is. 
     Second, since $\left(\xRet, \yRet, \mRet\right)$ is a solution for \ref{eq:app-vsums-OP}, $\xRet$ satisfies Constraint (\ref{eq:app-vsums-OP}.1) which is similar to Constraint (\ref{eq:min-sum-OP}.1).
     In addition, it means that $\left(\xRet, \yRet, \mRet\right)$ satisfies Constraints (\ref{eq:app-vsums-OP}.2--4). 
     By \Cref{lemma:linear-const-eqv}, this means that $\xRet$ satisfies Constraint (\ref{eq:min-sum-OP}.2).
     Thus, $\xRet$ is a solution for \ref{eq:min-sum-OP}.}
     
     \rmark{(b) Suppose by contradiction that $\xRet$ is not $\frac{1}{\multApprox}$-approximately-optimal. As \ref{eq:min-sum-OP} is a minimization program, this means that there exists a solution $\x$ for \ref{eq:min-sum-OP} such that $\objPofV{\ref{eq:min-sum-OP}}{\x} > \frac{1}{\multApprox}\ \objPofV{\ref{eq:min-sum-OP}}{\xRet}$.
     By \Cref{lemma:linear-const-eqv}, this means that there exist $\mathbf{y}$ and $\mathbf{m}$ such that $\left(\x, \mathbf{y}, \mathbf{m}\right)$ is a solution for \ref{eq:app-vsums-OP}. However, as both \ref{eq:min-sum-OP} and \ref{eq:app-vsums-OP} has the same objective, we get that:
     \begin{align*}
         \objPofV{\ref{eq:app-vsums-OP}}{\left(\xRet, \yRet, \mRet\right)} > \frac{1}{\multApprox} \objPofV{\ref{eq:app-vsums-OP}}{\left(\x, \mathbf{y}, \mathbf{m}\right)}
     \end{align*}
     In contradiction to the fact that $\left(\xRet, \yRet, \mRet\right)$ is $\frac{1}{\multApprox}$-approximately-optimal for  \ref{eq:app-vsums-OP}.}
\end{proof}

%% file: appendices/primal-dual.tex
\section{Primal-Dual Derivation}\label{apx:primal-dual}

This is the primal LP - the program \ref{eq:app-vsums-OP}, in standard form, with the corresponding dual variable shown to the left of each constraint.

\begin{align}
    &\min \quad \sum_{j=1}^{|S|} x_{j} \quad s.t. \tag{\progAppSecond}\\
    q_{\ell}& (1) \Hquad \ell y_{\ell} - \sum_{i=1}^n m_{\ell,i}\geq \sum_{i=1}^{\ell}  z_i && \forXinY{\ell}{t} \nonumber \\
    v_{\ell,i}& (2) \Hquad m_{\ell,i} - y_{\ell} + \sum_{j=1}^{|S|} x_j \cdot u_i(s_j) \geq 0  && \forXinY{\ell}{t},\Hquad \forXinY{i}{n} \nonumber \\
    & (3) \Hquad m_{\ell,i} \geq 0  && \forXinY{\ell}{t},\Hquad \forXinY{i}{n} \nonumber\\
    & (4) \Hquad x_{j} \geq 0 &&  j = 1, \ldots, |S| \nonumber
\end{align}

This is the dual LP - the program \ref{eq:dual-vsums-OP}, in standard form, with the corresponding primal variable shown to the left of each constraint.
\begin{align}
& \max \quad  \sum_{\ell=1}^{t} q_{\ell} \sum_{i=1}^{\ell}z_i  \quad s.t. \tag{\dualApp}\\
&\begin{aligned}
    x_j& (1) \Hquad \sum_{i=1}^n u_i(s_j) \sum_{\ell=1}^{t} v_{\ell,i}\leq 1 && \forall j = 1, \ldots, |S| \nonumber \\
 y_{\ell}& (2) \Hquad \ell q_{\ell} - \sum_{i=1}^n v_{\ell,i} \leq 0 && \forXinY{\ell}{t} \nonumber \\
m_{\ell,i} & (3) \Hquad -q_{\ell} +v_{\ell,i} \leq 0 && \forXinY{\ell}{t},\Hquad \forXinY{i}{n} \nonumber \\
& (4) \Hquad q_{\ell} \geq 0  &&\forXinY{\ell}{t} \nonumber\\
& (5) \Hquad v_{\ell,i} \geq 0  && \forXinY{\ell}{t},\Hquad \forXinY{i}{n} \nonumber
\end{aligned}
\end{align}

%% file: appendices/ellipsoid.tex
\section{Ellipsoid Method Variant for Approximation}\label{apx:ellipsoid}
This appendix presents a variant of the ellipsoid method designed to approximate linear programs (LPs) that cannot be solved directly due to a large number of variables. The method relies on an approximate separation oracle for the dual program.
The appendix uses standard notation for linear programs (both primal and dual); it is self-contained, and the notations used here are independent of the notation used in the main paper.
The method integrates techniques from  \cite{grotschel_geometric_1993,grotschel_ellipsoid_1981,karmarkar_efficient_1982}.


\begin{lemma}\label{lemma:ellipsoid-with-approx-oracle}
    Given a $\frac{1}{\multApprox}$-approximate-separation-oracle for the (max.) dual program, ~a poly-sparse $\frac{1}{\multApprox}$-approximately-optimal solution for the (min.) primal program can be obtained in polynomial time.
\end{lemma}

The goal is to solve the following linear program (the primal):
\begin{align}
\tag{P}
\begin{split}
\min \quad &c^T \cdot x \\
s.t. \quad &A \cdot x \geq b, \quad x\geq 0;
\end{split}
\end{align}
We assume that (P) has a small number of constraints, but may have a huge number of variables, so we cannot solve (P) directly. We consider its \emph{dual}:
\begin{align}
\tag{D}
\begin{split}
\max \quad & b^T \cdot y \\
s.t. \quad &A^T \cdot y \leq c,\quad y\geq 0.
\end{split}
\end{align}
Assume that both problems have optimal solutions and denote the optimal solutions of (P) and (D) by $x^{*}$ and $y^{*}$ respectively. By the strong duality theorem:
\begin{align}
    c^T \cdot x^{*} = b^T \cdot y^{*}
\end{align}

While (D) has a small number of variables, it has a huge number of constraints, so
we cannot solve it directly either. 
In this Appendix, we show that (P) can be approximated using the following tool:

\begin{definition}
An \emph{approximate separation oracle} (ASO) for the dual LP is an efficient function parameterized by a constant $\multError \geq 0$.
Given a vector $y$  it returns one of the following two answers:
\begin{enumerate}
\item "$y$ is infeasible". In this case, is returns a violated constraint, that is, a row $a_i^T \in A^T$ such that $a_i^T  y > c_i$.
\item "$y$ is \emph{approximately feasible}". 
That means that $A^T y \leq (1+\multError) \cdot c$
\end{enumerate}

\end{definition}
Given the ASO, we apply the ellipsoid method as follows (this is just a sketch
to illustrate the way we use the ASO; it omits some technical details):
\begin{itemize}
    \item Let $E_0$ be a large ellipsoid, that contains the entire feasible region, that is, all $y \geq 0$ for which $A^T y \leq c$.

    \item For $k = 0,1,\dots, K$ (where $K$ is a fixed constant, as will be explained later):
    \begin{itemize}
        \item Let $y_k$ be the centroid of ellipsoid $E_k$.
        
        \item Run the ASO on $y_k$.
        
        \item If the ASO returns "$y_k$ is infeasible" and a violated constraint $a_i^T$, then make a \emph{feasibility cut} --- keep in $E_{k+1}$ only those $y \in E_k$ for which $a_i^T y \leq c_i$.
        
        \item If the ASO returns "$y$ is approximately feasible", then make an \emph{optimality cut} --- keep in $E_{k+1}$ only those $y \in E_k$ for which $b^T y \geq b^T y_k$.
    \end{itemize}
    
    \item From the set $y_0, y_1, \dots, y_K$, choose the point with the highest $b^T \cdot y_k$ among all the approximately-feasible points.
\end{itemize}
Since both cuts are through the center of the ellipsoid, the ellipsoid dilates by a factor of at least $\frac{1}{r}$ at each iteration, where $r > 1$ is some constant (see \cite{grotschel_ellipsoid_1981} for computation of $r$). Therefore, by choosing $K := \log_2 r \cdot L$, where $L$ is the
number of bits in the binary representation of the input, the last ellipsoid $E_K$ is so small that all points in it can be considered equal (up to the accuracy of the binary representation).

The solution $y'$ returned by the above algorithm satisfies the following two conditions:
\begin{equation} \label{mult:y-star-is-approximetly-feasible}
     A^T y' \leq (1+\multError)\cdot c
\end{equation}
\begin{equation} \label{mult:y-star-obj-geq-opt}
     b^T y' \geq b^T y^{*}
\end{equation}
Inequality \eqref{mult:y-star-is-approximetly-feasible} holds since, by definition, $y'$ is approximately-feasible.

To prove \eqref{mult:y-star-obj-geq-opt}, suppose by contradiction that $b^T y^{*} > b^T y'$. 
Since $y^{*}$ is feasible for (D), it is in the initial ellipsoid. 
It remains in the ellipsoid throughout the algorithm: it is removed neither by a feasibility cut (since it is
feasible), nor by an optimality cut (since its value is at least as large as all values used for optimality cuts).
Therefore, it remains in the final ellipsoid, and it is chosen as the highest-valued feasible point rather than $y'$ --- a contradiction.

Now, we construct a reduced version of (D), where there are only at most $K$ constraints --- only the constraints used to make feasibility cuts.
Denote the reduced constraints by $A_{red}^T \cdot y \leq c_{red}$, where $A_{red}^T$ is a matrix containing a subset of at most $K$ rows of of $A^T$, and $c_{red}$ is a vector containing the corresponding subset of the elements of $c$. The reduced-dual LP is:
\begin{equation}
\tag{RD}
\begin{split}
\max  \quad & b^T y \\
s.t. \quad & A_{red}^T \cdot y \leq c_{red}, \quad y\geq 0
\end{split}
\end{equation}
Notice that it has the same number of variables as the program (D). Further, if we had run this ellipsoid method variant on (RD) (instead of (D)), then the result would have been exactly the same --- $y'$.
Therefore, \eqref{mult:y-star-obj-geq-opt} holds for the (RD) too:
\begin{equation} \label{mult:y-star-to-y-redopt}
    b^T y' \geq b^T y^{*}_{red}
\end{equation}
where $y^{*}_{red}$ is the optimal value of (RD).

As $A_{red}^T$ contains a subset of at most $K$ rows of $A^T$, the matrix $A_{red}$ contains a subset of \emph{columns} of $A$.
Therefore, the dual of (RD) has only at most $K$ variables, which are those who correspond to the columns of $A_{red}$:\footnote{Each column of $A$ is associated with a variable of the primal (P).  The variables of (RP) are those who are associated with the columns of $A$ that remain after the reduction process from $A$ to $A_{red}$.}
\begin{equation}
	\tag{RP}
    \begin{split}
     \min \quad &c_{red}^T \cdot x_{red} \\
            s.t. \quad &A_{red} \cdot x_{red} \geq b, \quad x_{red}\geq 0
    \end{split}
\end{equation}
Since (RP) has a polynomial number of variables  and constraints, it can be solved exactly by any LP solver (not necessarily the ellipsoid method).
Denote the optimal solution by $x^{*}_{red}$. 

Let $x'$ be a vector which describes an assignment to the variables of (P), in which all variables that exist in (RP) have the same value as in $x^{*}_{red}$, and all other variables are set to $0$.
It follows that $A \cdot x' = A_{red} \cdot x^{*}_{red}$, therefore, since $x^{*}_{red}$ is feasible for (RP), $x'$ is feasible for (P).
Similarly, $c^T \cdot x' = c^T_{red} \cdot x^{*}_{red}$.
We shall now see that this implies that the objective obtained by $x'$ approximates the objective obtained by $x^{*}$:

\begin{align*} 
    &c^T \cdot x' = c^T_{red} \cdot x^{*}_{red} \\
& \equWithExp{\text{By strong duality for the reduced LPs}}{=  b^T \cdot y^{*}_{red}}\\
 & \equWithExp{\text{By Equation (\ref{mult:y-star-to-y-redopt})}}{\leq  b^T\cdot y'}\\
 & \equWithExp{\text{By the definition of (P)}}{ \leq  (A \cdot x^{*})^T y'}\\
& \text{By properties of transpose}\\
&\quad  \text{and associativity of multiplication:} \\
& =  (x^{*})^T (A^T\cdot y') \\
 & \equWithExp{\text{By Equation \eqref{mult:y-star-is-approximetly-feasible}}}{\leq  (x^{*})^T ((1+\multError)\cdot c)}\\
& \equWithExp{\text{By properties of transpose}}{= (1+\multError) \cdot (c^T x^{*})}
\end{align*}
So, $x'$ ($x^{*}_{red}$ with all missing variables set to $0$) is an approximate solution to the primal LP (P) --- as required.

Notice that the number of entries with a non-zero value in $x'$ is at most polynomial, and therefore it is a poly-sparse.

\subsection{Using a Randomized Approximate Separation Oracles}\label{sec:random-sep-oracle}
Here, we allow the oracle to also be \emph{half-randomized}, that is, when it says that a solution is infeasible, it is always correct \rmark{and returns a violated constraint}; however, when it says that a solution is approximately feasible, it is only correct with some probability $p \in [0,1]$.
\rmark{Let $E$ be an upper bound on the number of iteration of the ellipsoid method on the given input when operating with a deterministic oracle; we prove that a half-randomized oracle can be utilized as follows:}
\begin{lemma}\label{lemma:random-ellipsoid}
    Given a \emph{half-randomized} $\frac{1}{\multApprox}$-approximate-separation-oracle for the (max.) dual program, with success probability $p \in (0,1]$,
    ~~a poly-sparse $\frac{1}{\multApprox}$-approximately-optimal solution for the (min.) primal program can be obtained in polynomial time with probability $p^{E}$.
\end{lemma}

\begin{proof}

    Since the ellipsoid method variant is iterative, and since the oracle calls are independent, there is a probability $p^E$ that the oracle answers correctly in each iteration, and so, the overall process performs as before. 
    We first explain why, using a half-randomized oracle, this ellipsoid method variant \emph{always} returns a feasible solution to the primal (even if the oracle was incorrect).
    Since the oracle is always correct when it determines that a solution is infeasible and as the construction of (RD) is entirely determined by the violated constraints, we can use the same arguments to conclude that $x'$ would still be a feasible solution for P.

    However, since the oracle might be mistaken when it determines that a solution is approximately-feasible for the dual, ellipsoid method variant might return a solution that not necessarily have an approximately optimal objective value.
    On the other hand, if the oracle is correct in all of its operations, the ellipsoid method variant would indeed produce an approximately optimal solution.
    That is, with probability $p^E$ the ellipsoid method variant returns an approximately optimal solution (as was for the deterministic oracle). 

\end{proof}

%% file: appendices/randomized-black-box.tex
\section{Using a Randomized Black-box ~~ (Proof of \Cref{lemma:randomized-solvers})}\label{apx:random}

This appendix extends our main result to the use of a \emph{randomized} black-box for utilitarian welfare, defined as follows: the black-box returns a selection that $\multApprox$-approximates the optimal utilitarian welfare with probability $p$, and an arbitrary selection otherwise.
\rmark{Section \ref{sec:random-req-bg} provides a summary of known claims we use inside the proof.} 

Recall that \Cref{lemma:randomized-solvers} says that:

\begin{theorem*}
    Given a randomized $\multApprox$-approximate black-box for the utilitarian welfare with success probability $p \in (0,1]$. ~~
    An $\multApprox$-leximin-approximation
    can be obtained with probability $p$ in time polynomial in $n$ and the running time of the black-box. 
\end{theorem*}

\begin{proof}
    \rmark{We first prove that the use of the randomized black-box does not effect feasibility --- that is, the output returned by Algorithm \ref{alg:basic-ordered-Outcomes} is always an \dist. Then, we prove the required guarantees regarding its optimally.}
    
    Recall that the black-box is used only in two places -- as part of the binary search to obtain an upper bound and as part of the separation oracle for \ref{eq:dual-vsums-OP}.
    
    Inside the binary search, if the obtained value does not approximate the optimal utilitarian welfare, it might cause us to overlook larger possible objective values. 
    While this could affect optimality, it will not impact feasibility.
    
    As for the separation oracle, this change makes the oracle \emph{half-randomized} as described in \Cref{sec:random-sep-oracle} since we might determine that Constraint (\dualApp.1) is approximately-feasible even though it is not. 
    However, by \Cref{lemma:random-ellipsoid}, the solution returned by the ellipsoid method will still be feasible for the primal, which makes the solution we eventually return through Algorithm \ref{alg:basic-ordered-Outcomes} also feasible (i.e., in $X$).

    We can now move on to the optimally guarantees. 
    Let $k$ be an upper bound on the total number of calls for the black-box.
    \begin{align*}
        k = n \cdot (\underset{\text{for upper bound}}{\underbrace{1}} + \underset{\text{binary search}}{\underbrace{\log{U}}}\cdot \underset{\text{ellipsoid}}{\underbrace{E}})
    \end{align*}
    Notice that $k$ is polynomial in the problem size.
    
    Under the assumption that success events of different activations are independent, it is clear that if the solver succeeds in all the executions, the returned solution will be a leximin approximation (as everything behaves as it would with a deterministic solver). 
    Therefore, with probability $p^k$, we return a leximin approximation.
    \rmark{However, $p^{k} < p$ since $k >1$ and $p \in (0,1)$.}
    
    To increase the probability of success, we can use a new black-box that operates as follows. It repeatedly calls the given black-box and then returns the best outcome.
    The probability of success of this new black-box is the probability that at least one of these operations is successful, which can also be calculated as $1$ minus the probability that all operation fail.
    Specifically, by operating the giving (original) black-box $q \geq 1$ times, the success probability is $(1-(1-p)^q)$.
    We then get that the success probability of the overall success probability of the algorithm becomes $\left( 1-(1-p)^q\right)^k$, which is at least $p$ for $q:=  \lceil \frac{\log{ (1/k)}}{\log{(1-p)}} + 1 \rceil$ :
    \begin{align*}
        &\left( 1-(1-p)^q\right)^k \\
        &\geq 1-k \cdot (1-p)^q  && \text{(By \Cref{claim:poly-bound-1})}\\
        &= 1-k \cdot (1-p)^{\lceil \frac{\log{ (1/k)}}{\log{(1-p)}} + 1 \rceil} \\
        &\geq 1-k \cdot (1-p)^{\frac{\log{ (1/k)}}{\log{(1-p)}} + 1} &&\text{(By \Cref{claim:poly-bound-2})}\\ 
        &= 1-k \cdot (1-p)^{\frac{\log{ (1/k)}}{\log{(1-p)}} } (1-p)\\
        &= 1-k \cdot \frac{1}{k} \cdot (1-p) = p && \text{(By \Cref{claim:poly-bound-3})}
    \end{align*}

    It is important to note that $q$ is $\mathcal{O}(\log{k})$:
    \begin{align*}
        q&=  \lceil \frac{\log{ (1/k)}}{\log{(1-p)}} + 1 \rceil \leq \frac{\log{ (1/k)}}{\log{(1-p)}} + 2\\
        &=\frac{\log{ k}}{\log{\frac{1}{(1-p)}}}+2
    \end{align*}
    
    As $k$ is polynomial in the problem size, \Cref{alg:basic-ordered-Outcomes} remains polynomial in the problem size even when using this new black-box.




    

\end{proof}

\subsection{Required Background}\label{sec:random-req-bg}

\begin{claim}\label{claim:poly-bound-1}
    For any $\epsilon \in (0,1)$ and $k \in \mathbb{Z}_{+}$:
    \begin{align*}
        (1 - \epsilon)^k \geq 1 - k \cdot \epsilon
    \end{align*}
\end{claim}

\begin{claim}\label{claim:poly-bound-2}
For any $a \in (0,1)$, and $c>b>0$:
    \begin{align*}
        a^c \leq a^b
    \end{align*}
\end{claim}

\begin{claim}\label{claim:poly-bound-3}
For any $a,b >0$:
    \begin{align*}
        b^{\frac{\log{a}}{\log{b}}} = a
    \end{align*}
\end{claim}


%% file: appendices/approx-def-comp.tex
\section{Definitions for Leximin Approximation: Comparison}\label{apx:def-lex-approx}

In this paper, we propose a new definition for leximin-approximation that applies only problems considered here (for which using a lottery is meaningful). 
There are other definitions for leximin approximation. 
This appendix provides an extensive comparison between the different types of definitions for leximin-approximation.
However, to make comparison easier, we describe the definitions only for our model.

To illustrate the difference, we use a slightly modified version of an example given by \cite{abernethy2024lexicographic} (Appendix B)\footnote{The definition considered in \cite{abernethy2024lexicographic} is for additive approximation, we adapt it all for multiplicative approximations.} to compare their definition to the one by \citet{hartman2023leximin}.
\begin{example}\label{example:for-def}
    We have $3$ agents and only $4$ possible outcomes in $X$, where:
    \begin{align*}
        &\mathbf{E}^{\uparrow}(\x^1) = (10, 10, 100)\\
        &\mathbf{E}^{\uparrow}(\x^2) = (9, 9, 90)\\
        &\mathbf{E}^{\uparrow}(\x^3) = (9, 50, 50)\\
        &\mathbf{E}^{\uparrow}(\x^4) = (8, 1000, 1000)
    \end{align*}
\end{example}

For simplicity, we consider $\multApprox = 0.9$.
Notice that the optimal solution is $\x^1$.

\paragraph{Our Definition.}
Recall that our definition says that an \dist $\xt$ is $\multApprox$-leximin-approximation if $\mathbf{E}(\xt) \weaklyPreferred \multApprox \mathbf{E}(\x)$ for any $x \in X$.

In Example \ref{example:for-def}, our definition says that the three solutions $\x^1, \x^2$, and $\x^3$ are $0.9$-approximations. However, $x_4$ is not since $0.9 \cdot \mathbf{E}(\x^1) \succ \mathbf{E}(\x^4)$.

\paragraph{\citet{hartman2023leximin}.}
Recently, \citet{hartman2023leximin} proposed a general definition for leximin-approximation,\footnote{The definition in \cite{hartman2023leximin} capture both -additive and multiplicative errors, combined; here we only consider multiplicative errors.} which in our context can be described as follows. 
First, an \dist $\x$ is said to be $\multApprox$-preferred over another \dist $\xt$, denoted by $\x \succ_{\multApprox} \xt$, if there exists an $1 \leq k \leq n$ such that $\expectedValBy{i}{\x} \geq \expectedValBy{i}{\xt}$ for $i < k$ and $\expectedValBy{k}{\x} > \frac{1}{\multApprox}\expectedValBy{k}{\xt}$.
Then, an \dist $\xt$ is said to be $\multApprox$-leximin-approximation if there is \emph{no} $\x$ such that $\x \succ_{\multApprox} \xt$.

In Example \ref{example:for-def}, their definition says that the two solutions $\x_1$ and $\x_3$ are $0.9$-approximations. However, both $\x^2$ and $\x^4$ are not - $\x^2$ is not since $\x^3 \succ_{0.9} \x^2$ (for $k=2$), while $\x^4$ is not since $\x^1 \succ_{0.9} \x^4$  (for $k=1$).

This definition is stronger than ours:
\rmark{
\begin{claim}
    Any $\multApprox$-leximin-approximation according to the definition of \citet{hartman2023leximin} is also a $\multApprox$-leximin-approximation according to our definition, but the opposite is not true. 
\end{claim}
}
\begin{proof}
    \rmark{Let $\xt \in X$. We first prove that if $\xt$ is not an $\multApprox$-leximin-approximation according to our definition then $\xt$ is not an $\multApprox$-leximin-approximation according to their definition. 
    Clearly, this implies that if $\xt$ is an approximation according to their definition, it is also an approximation according to our definition.
    If $\xt$ is not an approximation according to our definition, this means that there exists an $\x \in X$ such that $\multApprox \mathbf{E}(\x) \leximinPreferred \mathbf{E}(\xt)$. Which means that there exists an integer $k \in [n]$ such that $\multApprox \expectedValBy{i}{\x} = \expectedValBy{i}{\xt}$ for $i<k$ and $\multApprox\expectedValBy{k}{\x} > \expectedValBy{k}{\xt}$.
    For any $i<k$, as $\multApprox \in (0,1]$, we get that $\multApprox \expectedValBy{i}{\x}  = \expectedValBy{i}{\xt}$ implies $\expectedValBy{i}{\x}  \geq  \expectedValBy{i}{\xt}$.
    For $k$, it is easy to see that $E_k(\x) > \frac{1}{\multApprox} E_k(\xt)$.
    Therefore, $\x \succ_{\multApprox} \xt$. But this means that $\xt$ is not an approximation according to their definition.}

    \rmark{To see that the opposite is not true, consider \Cref{example:for-def} and observe that $\x^2$ is an approximation according to our definition but it is not according to their.}
\end{proof}
However, this paper proves that our definition can be obtained for many problems, whereas their approximation appears to be very challenging to obtain.

\paragraph{Element-Wise.} 
Another common definition is an element-wise approximation, e.g., in \cite{abernethy2024lexicographic} and \cite{Kleinberg20012Routing} (in which it is called a coordinate-wise approximation). In our context, it can be described as follows. An \dist $\xt$ is said to be $\multApprox$-leximin definition if $\expectedValBy{i}{\x} \geq \multApprox \expectedValBy{i}{\x^*}$ for all $i \in N$, where $\x^*$ is the leximin-optimal solution. 

In Example \ref{example:for-def}, this definition says that the two solutions $\x_1$ and $\x_2$ are $0.9$-approximations. However, both $\x^3$ and $\x^4$ are not - $\x^3$ is not since $\x^3_3 = 50 < 90 = 0.9 \cdot 100 = 0.9 \cdot \x^2_2$, while $\x^4$ is not since $\x^*_1 = 8 < 9 = 0.9 \cdot 10 = 0.9 \cdot \x^*_1$ (recall that $\x^* = \x^1$ in this example).

This definition is also stronger than ours in the same sense (as described for \cite{hartman2023leximin}):
\rmark{
\begin{claim}
    Any $\multApprox$-leximin-approximation according to the element-wise definition is also a $\multApprox$-leximin-approximation according to our definition, but the opposite is not true. 
\end{claim}
}
\begin{proof}
    \rmark{Let $\x^* \in X$ be a leximin-optimal \dist and let $\xt \in X$ be an $\multApprox$-leximin-approximation according to the element-wise definition.
    This means that $\expectedValBy{i}{\xt} \geq \multApprox \expectedValBy{i}{\x^*}$ for any $i \in [n]$.
    We first prove that this also means that $\mathbf{E}(\xt) \weaklyPreferred \multApprox \mathbf{E}(\x^*)$.
    If $\mathbf{E}(\xt) \equiv \multApprox \mathbf{E}(\x^*)$ then the claim is clearly holds. 
    Otherwise, let $k$ be the smallest integer such that $\expectedValBy{k}{\xt} > \multApprox \expectedValBy{k}{\x^*}$ (notice that there must be one since $\expectedValBy{i}{\xt} \geq \multApprox \expectedValBy{i}{\x^*}$ for any $i \in [n]$ and $\mathbf{E}(\xt) \not\equiv \multApprox \mathbf{E}(\x^*)$).
    We get that $ \expectedValBy{i}{\xt} = \expectedValBy{i}{\x^*}$ for any $i<k$ and $\expectedValBy{k}{\xt} > \multApprox \expectedValBy{k}{\x^*}$.
    Therefore, $\mathbf{E}(\xt) \weaklyPreferred \multApprox \mathbf{E}(\x^*)$.
    Now, let $\x \in X$. 
    By the optimality of $\x^*$, $\mathbf{E}(\x^*) \weaklyPreferred \mathbf{E}(\x)$.
    By transitivity,
    $\mathbf{E}(\xt) \weaklyPreferred \multApprox \mathbf{E}(\x)$ --- as required.}

    \rmark{To see that the opposite is not true, consider \Cref{example:for-def} and observe that $\x^3$ is an approximation according to our definition but it is not according to their.}
\end{proof}

We note that here, our definition captures the following problem. 
The solution $\x^2$ is considered an element-wise $0.9$-approximation while $\x^3$ does not. However, the expected vector of $\x^3$ is strongly-leximin-preferred over the expected vector of $\x^2$ (for $k=2$).
Thus, in the leximin sense, it seems reasonable that $\x^3$ would be considered at-least-a-good as $\x^2$, yet this is not reflected under the element-wise definition. 
According to our definition, $\x^3$ is also considered a leximin-approximation, and this is no coincidence — any solution preferred over another that is an approximation is itself considered an approximation.